\documentclass[twocolumn,superscriptaddress,aps]{revtex4-2}

\usepackage[utf8]{inputenc}
\usepackage{url}
\usepackage{amssymb,amsmath,amsthm}
\usepackage{bm}
\usepackage{graphicx}
\usepackage{enumerate}
\usepackage{microtype}
\usepackage{float}
\makeatletter
\IfFileExists{placeins.sty}{\usepackage{placeins}}{
  \newcommand\FloatBarrier{\par\begingroup\let\@elt\relax
    \edef\@tempa{\@botlist\@deferlist\@dbldeferlist}%
    \ifx\@tempa\@empty\else\clearpage\fi\endgroup}}
\makeatother
\usepackage{color}
\usepackage{hyperref}
\usepackage{multirow}
\usepackage{array}
\newcolumntype{C}[1]{>{\centering\let\newline\\\arraybackslash\hspace{0pt}}m{#1}}
\graphicspath{{figures/}}

\usepackage[ruled,linesnumbered,vlined]{algorithm2e}

\newtheorem{defn}{\noindent $\mathbf{Definition}$}[section]

\newtheorem{theorem}[defn]{$\mathbf{Theorem}$}
\newtheorem{cor}[defn]{$\mathbf{Corollary}$}

\begin{document}

\title{A unified geometric design framework for kirigami structures}
\author{Qinghai Jiang}
\affiliation{Department of Mathematics, The Chinese University of Hong Kong}
\author{Gary P. T. Choi}
\thanks{To whom correspondence may be addressed. Email: ptchoi@cuhk.edu.hk.}
\affiliation{Department of Mathematics, The Chinese University of Hong Kong}


\begin{abstract}
Kirigami turns a thin sheet with cuts into a shape-morphing mechanical metamaterial, with applications across science and engineering. While various two- and three-dimensional kirigami design methods have been developed, most of them are only applicable to a limited class of kirigami structures. In this work, we present a unified geometric design framework for kirigami structures. Since the states of a kirigami pattern differ only in which tile corners coincide, one tessellation routine generates them all, and edge-length equalities couple them. A single constrained optimization then designs 2D-to-2D, 2D-to-3D, and 3D-to-3D shape-morphing, compact reconfigurable, and rigidly deployable structures, and one sheet that deploys into two prescribed shapes. The designs further reveal geometric laws hidden in the kirigami mechanism, which explain how reconfiguration transposes the shape of a sheet, how one sheet bends into both a bowl and a saddle, and why the outline of a flat sheet limits the curvature it can reach when it closes again. Altogether, our work paves a new way for the design of shape-morphing mechanical metamaterials.

\end{abstract}

\maketitle

\section{Introduction}

Kirigami, the traditional art of paper cutting, has become a route to shape-morphing structures in science and engineering, from flexible electronics~\cite{jiang2022flexible} and robotics~\cite{yang2021grasping} to mesosurfaces~\cite{cheng2023programming} and parachutes~\cite{lamoureux2025kirigami}. Its cuts turn a sheet into a mechanism of tiles that rotate about shared corners, so that the sheet deploys into a new shape. Kirigami is thus part of a wider effort to program shape in mechanical metamaterials~\cite{bertoldi2017flexible,barchiesi2019mechanical,xia2022responsive,li2023auxetic}, where computational methods have been developed for auxetic surfaces~\cite{konakovic2016beyond,konakovic2018rapid,chen2021bistable}, inflatables~\cite{panetta2021computational}, many other planar and curved structures~\cite{nojoomi20212d,liu2021wallpaper,jiang2022shape,liu2022quasicrystal,zhai2021mechanical,jin2024engineering,toyonaga2026collapsible}, and 3D architected materials~\cite{shim2012buckling,haghpanah2016multistable,overvelde2017rational,li20243d,zaman2025one,gu2024kirigami,hong2025reprogrammable,giomi2026stretching,dang2025shape,segall2025reconfigurable}.

The inverse problem, to find cuts that carry a sheet between prescribed shapes, has been tackled for a growing list of cases: a 2D sheet that deploys into a prescribed 2D or 3D shape~\cite{choi2019programming}, a pattern with a second compact state~\cite{choi2021compact}, rigidly deployable patterns and other surface topologies~\cite{dang2021theorem,dang2022theorem}, and other tessellations and design goals~\cite{dudte2023additive,qiao2025inverse,segall2026uniformly}. Each case has come with its own formulation. Most describe a pattern by one deployed state and keep its tiles consistent through contractibility constraints, angle sums, and matched edges around every slit~\cite{choi2019programming}, and every further property adds constraints of its own. A shape change between two curved states, or one sheet that deploys onto several prescribed shapes, fits none of them. However, the highly case-specific formulations in these works make it difficult to apply and generalize them to other scenarios. Moreover, the underlying geometric laws and limits in the design problem are unclear.

\begin{figure}[t!]
    \centering
    \includegraphics[width=\linewidth]{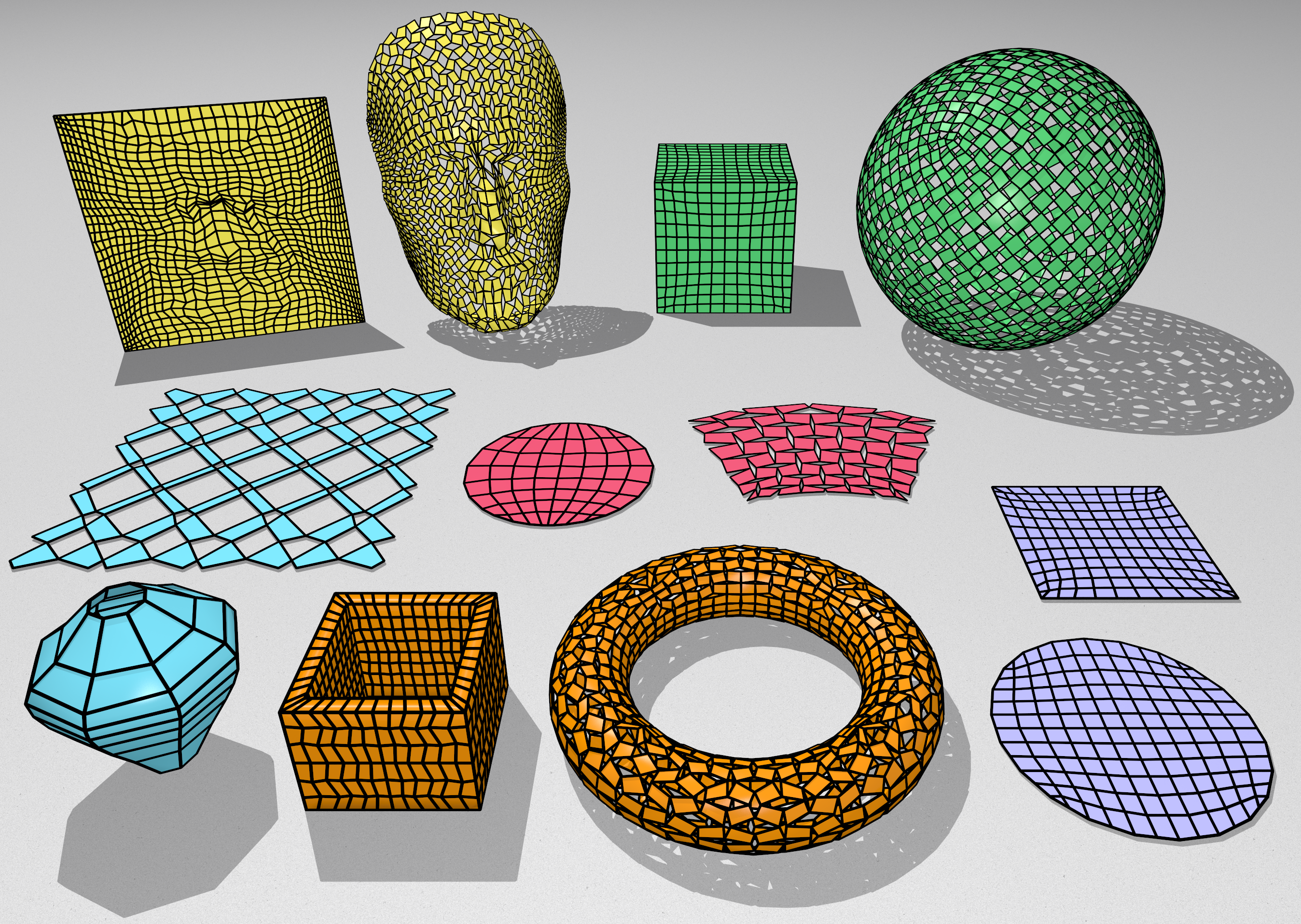}
    \caption{\textbf{A gallery of two- and three-dimensional shape-morphing kirigami structures created by our proposed unified design framework.} Each color represents a kirigami structure. For each structure, we show two configurations with the desired shape morphing effects: (Red) from a 2D compact circle to a 2D deployed rainbow shape;  (Purple) from a 2D compact square to 2D compact circle; (Yellow) from a 2D compact rectangle to a 3D human face shape; (Blue) from a 2D deployed rectangle to a 3D compact vase; (Green) from a 3D compact cube to a 3D deployed sphere; (Orange) from a 3D compact square ring to a 3D deployed torus.}
    \label{fig:F1}
\end{figure}

To address these issues, in this work we develop a unified framework for 2D and 3D kirigami design (Fig.~\ref{fig:F1}). It rests on one observation: the states of a kirigami pattern, compact or deployed, planar or curved, differ only in which tile corners coincide. One tessellation routine therefore generates every state, and equalities between the edge lengths of the same tile in two states couple them. Target shapes, compact reconfigurability, rigid deployability, and further states become constraints of one optimization problem. With it, we design 2D-to-2D, 2D-to-3D, and 3D-to-3D structures, and a single sheet that deploys into multiple 3D states with highly different target geometries (Section~\ref{sec:results}).

\begin{figure*}[t]
    \centering
    \includegraphics[width=\linewidth]{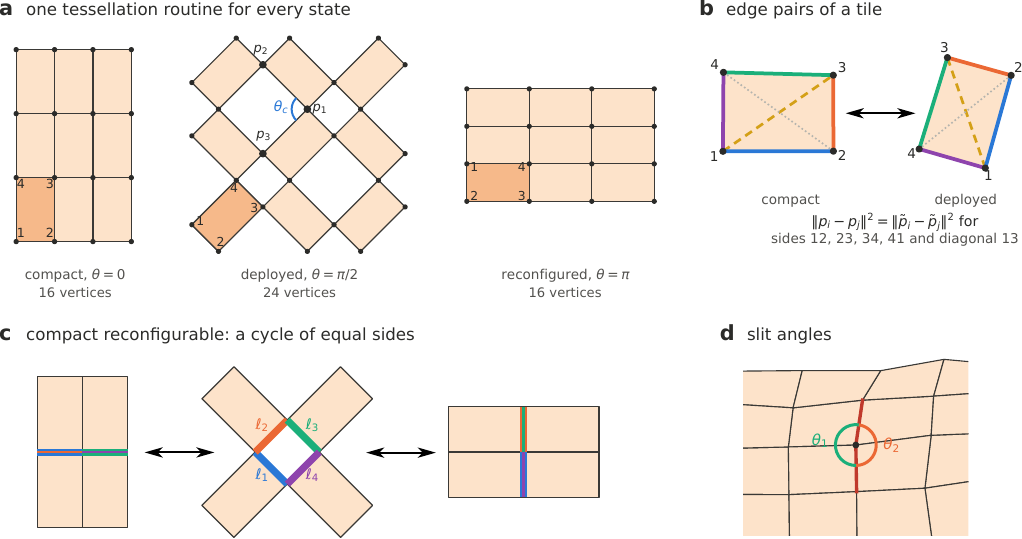}
    \caption{\textbf{The two ingredients of our framework: one tessellation routine and edge pairs.} \textbf{a},~The same routine generates every state of a pattern: the uniform pattern at cutting angle $\theta=0$ (compact; slits in red), $\theta=\pi/2$ (deployed) and $\theta=\pi$ (reconfigured compact), with the tile corners that coincide merged into one vertex (dots). Every tile keeps its corner labels $1$--$4$ in every state (shaded tile). In the deployed state, the cutting angle $\theta_c$ is formed by $p_1$, $p_2$ and $p_3$. \textbf{b},~Edge pairs: a tile of a designed pattern in the compact and in the deployed state; its four sides and the diagonal $13$ (one color per pair) are required to have equal lengths, Eq.~\eqref{edge-length-constraint}; the other diagonal (grey) follows (SI Appendix~\ref{appendix:KKT}). \textbf{c},~In a compact reconfigurable pattern the four sides of a negative space, of lengths $\ell_1,\dots,\ell_4$ (one color each), coincide in pairs one way in the compact state, $\ell_1=\ell_2$ and $\ell_3=\ell_4$, and the other way in the reconfigured state, $\ell_2=\ell_3$ and $\ell_4=\ell_1$ (overlapping strokes); the four equalities form a cycle, so one of them is implied (SI Appendix~\ref{appendix:reconfig-pairs}). \textbf{d},~The angles $\theta_1$ and $\theta_2$ on the two sides of the slit at an interior vertex of a designed compact pattern; the slit is straight when $\theta_1=\theta_2=\pi$.}
    \label{fig:tess_terminology}
\end{figure*}

Our unified design framework further reveals important geometric laws and limits in kirigami design. Specifically, since one framework produces every design, the designs can be compared across problems, and they share regularities that no constraint states. Compact reconfigurable patterns turn the shape they fill by a quarter turn; in the sheet that carries a bowl and a saddle, each group of tiles opens or closes with the sign of the curvature beneath it; and in a design that closes onto a curved surface, the curvature reappears as the turning of the planar compact outline. We trace all three to one lattice form of the cell law of Zheng et al.~\cite{zheng2022continuum} (Section~\ref{sec:analysis}): reclosure transposes the second moments of area of the two compact states, the curvature produced by a varying opening changes sign at the fold of the mechanism, and at reclosure the angle defects of the curved state are balanced, along chains of hinges, by the turning of the planar compact outline. Each law gives a design rule.

\section{Methods}
We consider kirigami structures of quadrilateral tiles with the standard rotating-squares topology, an $M\times N$ array in which each tile is connected at one corner to each of its four neighbors (Fig.~\ref{fig:tess_terminology}). The unknowns are the vertex coordinates of the pattern in \emph{both} the contracted and the deployed state. Two ingredients carry the whole formulation: one routine that produces every state of the pattern, and the edge pairs that tie the tiles of two states together. We describe them for 2D-to-2D design first and then extend them to 3D.

\subsection{One tessellation routine for every state}
The states of a pattern differ only in which tile corners coincide. We generate all of them with one routine: for a cutting angle $\theta$, it lays out the uniform pattern of $h\times w$ rectangles at the opening $\phi=\theta/2$, and merges the tile corners that fall on the same point into one vertex (Fig.~\ref{fig:tess_terminology}\textbf{a}). A uniform pattern at any opening follows from the compact one by a single deformation gradient: the lattice of tile centers is mapped by $A(\phi)=\cos\phi\,I+\sin\phi\,G$, with $G=\operatorname{diag}(h/w,\,w/h)$ for rectangles of width $w$ and height $h$, and every tile turns by $\pm\phi$ about its center~\cite{zheng2022continuum} (SI Appendix~\ref{appendix:uniform}). At $\theta=0$ this is the compact state, with $(M+1)(N+1)$ vertices; for $0<\theta<\pi$ a deployed state, with $(M+1)(N+1)+MN-1$ vertices; and at $\theta=\pi$ a second compact state, again with $(M+1)(N+1)$ vertices but merged differently. Every tile keeps its corner labels $1,\dots,4$ in every state, and the routine returns, for each state, the vertex at each labelled corner. Vertices shared in a state are thus single variables, and no coincidence constraint is needed. For an $N\times N$ pattern with a compact and a deployed state in 2D there are $2(N+1)^2+2(2N^2+2N)=6N^2+8N+2$ variables. The uniform patterns are also the initial guesses, each state at its own cutting angle, and need not be generated in the same way; a conformal or quasi-conformal map can bend the initial deployed pattern towards the target~\cite{choi2019programming}.

\subsection{Edge pairs}
Because the corner labels are shared, a tile in one state and the same tile in another state are tied by the pairs of segments with the same labels (Fig.~\ref{fig:tess_terminology}\textbf{b}). Each tile of the compact state should be congruent to its copy in the deployed state, and instead of an isometric transformation, we express this by \emph{edge-length equalities}
\begin{equation}
    \label{edge-length-constraint}
    e(x)=\|p_i - p_j\|^2 - \|\tilde{p}_i - \tilde{p}_j\|^2 = 0,
\end{equation}
where $p_i,p_j$ are corners of a tile in the compact state and $\tilde p_i,\tilde p_j$ the same corners in the deployed state. A quadrilateral is fixed up to a rigid motion by its four sides and its two diagonals; we impose the four sides and the diagonal between corners $1$ and $3$, because the other diagonal is then dependent on them near every feasible point (SI Appendix~\ref{appendix:KKT}). The pairs are read off the two corner maps with the same five local pairs in every tile, so the same few lines of code relate any two states.

Inequalities $a(x)\geq0$ exclude overlaps. For the cutting angle $\theta_c$ formed by $p_1,p_2,p_3$ (Fig.~\ref{fig:tess_terminology}\textbf{a}),
\begin{equation}
    a(x) = (p_2-p_1)\times (p_3-p_1) \geq 0,
\end{equation}
and likewise every tile angle stays in $(0,\pi)$. Prior formulations~\cite{choi2019programming,choi2021compact} optimize the deployed vertices only and enforce the consistency of the tiles through contractibility constraints, which include angle-sum equalities that are numerically delicate. Here, all equalities are edge lengths, with simple gradients, and the remaining angle conditions are inequalities, which do not hinder the interior-point method. The basic design problem is
\begin{equation}
\min_x f_0(x) \quad \text{subject to} \quad e_i(x)=0,\ \ a_j(x)\geq 0,
\end{equation}
where the $e_i$ are the edge-length equalities of all edge pairs and the $a_j$ the overlap and tile-angle inequalities, with an objective $f_0$ chosen by the user, solved by the interior-point method from the initial guesses above.

\subsection{Rigid deployability and rigid reconfigurability}
\label{sec:reconfig-rigid}
\emph{Rigid deployability.} A pattern without holes whose slits are straight in the compact state deploys with every tile rigid if and only if every negative space is a parallelogram~\cite{dang2021theorem}. We impose both conditions directly. As in the additive framework of Dudte et al.~\cite{dudte2023additive}, the vertex in the middle of each slit is split into two, $p_v$ and $p_v'$; the slit is straight, $\theta_1=\theta_2=\pi$ in Fig.~\ref{fig:tess_terminology}\textbf{d}, when both lie on the line through its ends, and its negative space is a parallelogram when one pair of its opposite sides has equal length (Fig.~\ref{fig:example_2D-to-2D}\textbf{b}). Every hinge then opens by the same angle as the pattern deploys.

\emph{Rigid reconfigurability.} A pattern is compact reconfigurable if it has a second compact state, which it reaches when all cutting angles reach $\pi$~\cite{choi2021compact} (Fig.~\ref{fig:tess_terminology}\textbf{a}, right). We generate this state with the same routine at $\theta=\pi$ and impose the edge pairs between the two compact states; unlike ref.~\cite{choi2021compact}, no contractibility constraint is needed. The edge pairs then repeat themselves. Around every interior vertex, the four sides of the negative space coincide in pairs one way in the first compact state and the other way in the second (Fig.~\ref{fig:tess_terminology}\textbf{c}), so their four equalities form a cycle and any three imply the fourth. The $5MN$ equalities therefore contain one dependency per interior vertex, at every configuration, and for $M=N\ge9$ they outnumber the $4(N+1)^2$ variables. We drop one side of every interior tile, which keeps their rank and leaves fewer equalities than variables; a compact rectangle needs a different choice (SI Appendix~\ref{appendix:reconfig-pairs}). The same cycle makes every negative space a rhombus, so a compact reconfigurable pattern whose slits are straight is rigidly deployable, and its single mechanism carries it from one compact state into the other; we call such a pattern \emph{rigidly reconfigurable}. Compact reconfigurability alone does not straighten the slits, but in all our 2D compact reconfigurable designs they are straight to solver accuracy, whether or not this is imposed, so these designs are rigidly reconfigurable. In compact reconfigurable 2D-to-3D designs the slits bend (Section~\ref{sec:limitation}).

\begin{figure*}[t!]
    \centering
    \includegraphics[width=\linewidth]{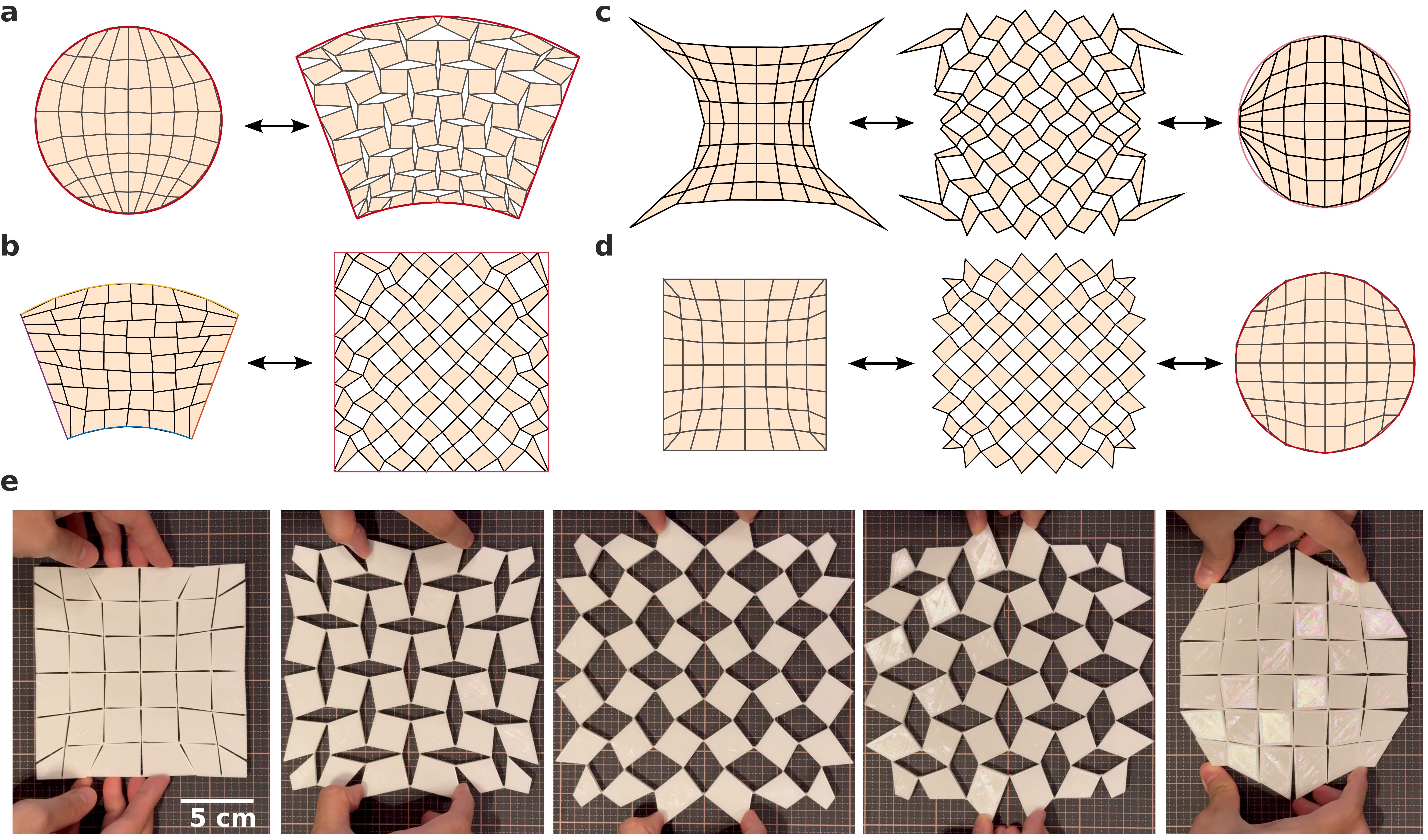}
    \caption{\textbf{Examples of 2D-to-2D kirigami structures designed by our framework.} \textbf{a}, A kirigami structure that morphs from a compact circle to a deployed state approximating a rainbow shape. \textbf{b}, A rigidly deployable kirigami structure with parallelogram negative spaces and a rainbow-to-rectangle shape change from compact to deployed states. \textbf{c}--\textbf{d}, Rigidly deployable and compact reconfigurable patterns with a circle as target shape in the second contracted state. For the compact pattern, one can set (\textbf{c}) a free boundary shape condition or (\textbf{d}) a square boundary condition and obtain different resulting designs. For each example, the left and right panels show the optimization results produced by our design framework, and the middle panel shows an intermediate deployed configuration of the pattern simulated by the PyKirigami simulator~\cite{jiang2025pykirigami}. \textbf{e}, A physical model of a rigidly deployable and compact reconfigurable square-to-circle kirigami structure. Scale bar = 5 cm.}
    \label{fig:example_2D-to-2D}
\end{figure*}

\subsection{Shape matching, 3D designs and several states}
\emph{Shape matching.} A straight side of a target is a linear equality on the boundary vertices, such as $y_i-y_j=0$ for a horizontal side. A curved target $\mathcal C$, a curve in 2D or a surface in 3D, enters through the distance of a vertex to it, either as an equality constraint, $d(p_i,\mathcal C)=0$, or in the objective, $f_0(x)=\sum_i\|p_i-q_i\|^2$ with $q_i$ the point of $\mathcal C$ closest to $p_i$; targets can be assigned to the compact state, the deployed state or both.

\emph{3D designs.} Nothing above depends on the dimension. In a 3D state the vertices have three coordinates and the edge pairs are unchanged, but the five lengths of a tile no longer prevent it from folding along its other diagonal; following~\cite{choi2019programming}, we require the tetrahedron spanned by its four corners to have zero volume. A 2D-to-3D compact reconfigurable design takes the $\theta=\pi$ state as before, with the vertices merged in 3D. A 3D-to-3D design can be assembled from 2D-to-3D patches whose boundary vertices are constrained to lie at consistent positions along the partition curves, or obtained directly from 3D initial guesses for both states.

\emph{Several states.} The states of a pattern interact only through the edge pairs with the shared compact state, so one sheet can be asked to reach several targets by stacking deployed states: for $K$ targets the unknowns are $x=[x_C;x_{D_1};\dots;x_{D_K}]$, each deployed state has its own edge pairs with the compact state and its own planarity, overlap and matching terms, and the objective sums the matching terms. Each state is initialized independently, at its own cutting angle.

\begin{figure*}[t]
        \centering
        \includegraphics[width=\linewidth]{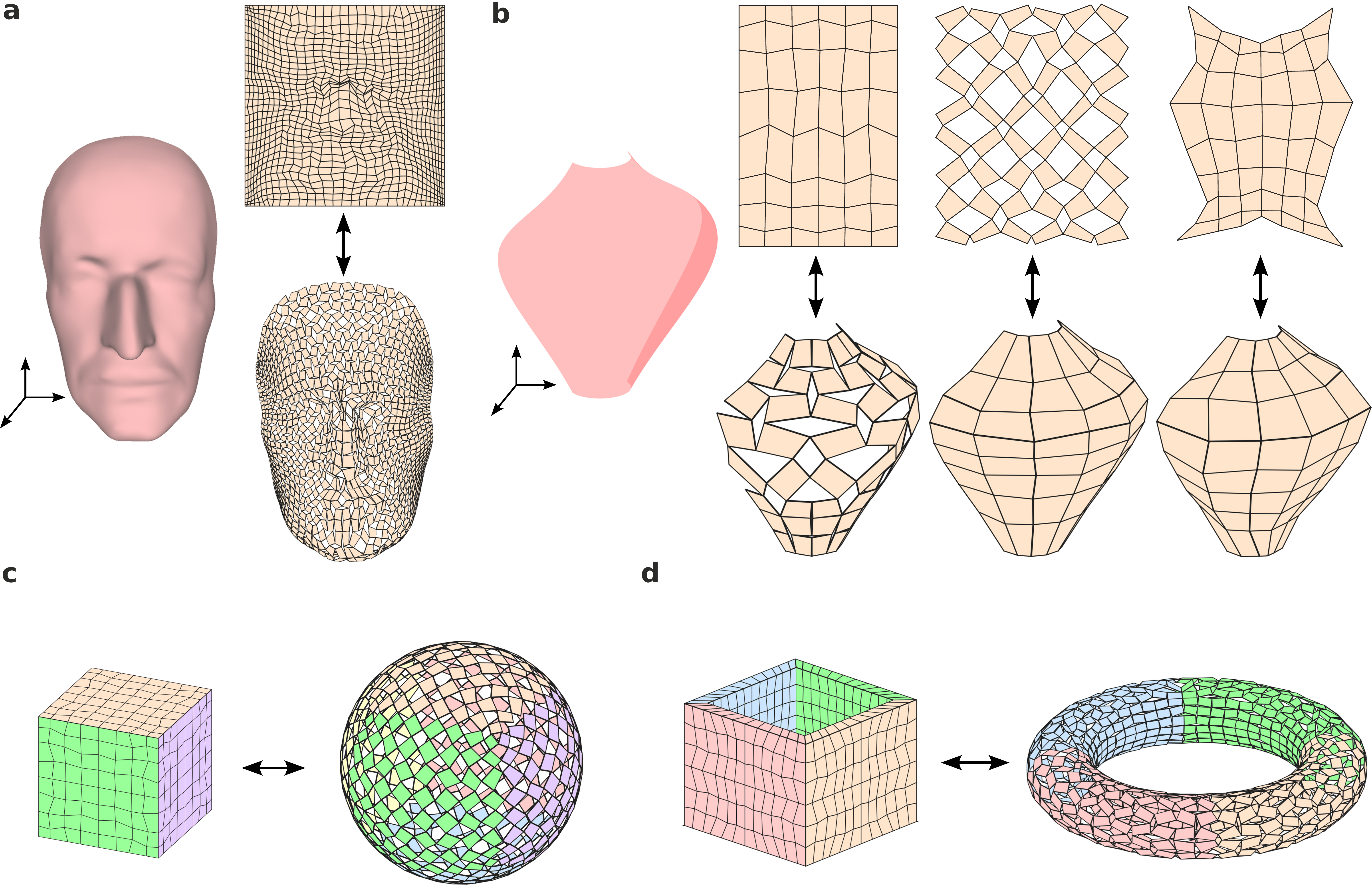}
        \caption{\textbf{Examples of 2D-to-3D and 3D-to-3D kirigami structures designed by our framework.} \textbf{a}, Given a target 3D Max Planck surface (left), we can design a kirigami structure that morphs from a 2D compact rectangle (top right) to a 3D deployed configuration matching the target surface (bottom right). \textbf{b}, We further consider matching a 3D half-vase surface with different conditions. First column: A structure that morphs from a 2D compact rectangle to a 3D deployed configuration matching the half-vase geometry. Second column: A structure that morphs from a 2D deployed rectangle to a 3D compact configuration matching the half-vase. Third column: A structure that morphs from a 2D contracted configuration to a 3D second contracted configuration approximating the half-vase shape. \textbf{c}, A cube-to-sphere kirigami structure designed by assembling six 2D-to-3D designs. \textbf{d}, A kirigami structure that morphs from a square ring to a torus.}
        \label{fig:examples-2D-to-3D}
\end{figure*}

\subsection{Implementation}
The framework is implemented in MATLAB, with a Python port used for the analysis in this work, and the constrained optimization problem is solved by the interior-point method in IPOPT~\cite{wachter2006implementation}. SI Appendix~\ref{appendix:optimization} gives the residuals and Jacobians of all constraints and explains which equalities are left out and why.

\section{Results} \label{sec:results}

Using our proposed unified design framework, we can easily design shape-morphing kirigami structures. Below, we highlight several representative examples of 2D-to-2D, 2D-to-3D, and 3D-to-3D kirigami structures created by our framework. More examples are provided in SI Appendix~\ref{appendix:additionalresults}.

\subsection{2D-to-2D kirigami structures}

In Fig.~\ref{fig:example_2D-to-2D}\textbf{a}, we consider designing a kirigami structure undergoing a circle-to-rainbow shape change. Specifically, it is desired that the kirigami structure forms a circle in the contracted state and then morphs into a rainbow shape in the deployed state. It can be observed that the optimization result satisfies the requirement perfectly.

Next, as described in our formulation, it is possible to enforce additional constraints in our framework for achieving rigidly deployable kirigami structures. In Fig.~\ref{fig:example_2D-to-2D}\textbf{b}, we show an example of a kirigami structure undergoing a rainbow-to-rectangle shape change from the contracted state to the deployed state, with every slit straight and every negative space a parallelogram. Both target shapes are matched to machine precision, and the deployed pattern has exactly one internal mechanism with rigid tiles, in which every hinge opens by the same angle ($0.25\pi$ at the designed state).

We then show examples of rigidly deployable and compact reconfigurable patterns (Fig.~\ref{fig:example_2D-to-2D}\textbf{c}--\textbf{d}). Here, we consider a circle as the target shape at the second contracted state of the pattern. As for the first compact pattern, we can either impose a free-boundary shape condition or a square boundary condition. In both cases, it can be observed that the optimization results produced by our proposed framework satisfy the requirements. In Fig.~\ref{fig:example_2D-to-2D}\textbf{e}, we further show the deployment and reconfiguration process of a physical model produced based on a rigidly deployable, compact reconfigurable square-to-circle kirigami structure designed by our framework. The physical model is fabricated using TPU (Thermoplastic Polyurethane) 3D printing, with a small portion between the corners of adjacent tiles kept as joints. It can be observed that the physical model can successfully undergo the desired square-to-circle transformation, illustrating the feasibility of our design.

\subsection{2D-to-3D kirigami structures}

Next, we apply our unified design framework to create kirigami structures with prescribed 2D-to-3D shape-morphing effects. In Fig.~\ref{fig:examples-2D-to-3D}\textbf{a}, we consider a 3D human face as the target shape to be approximated by the kirigami structure. We also enforce that the contracted 2D shape is a rectangle. Under our optimization framework, we obtain a kirigami structure that achieves the desired effects. It is noteworthy that the tile shapes are not uniform. In particular, the optimized tiles that approximate the nose bridge region are more elongated to conform to the narrow, sharp geometry. This demonstrates the ability of our design framework to handle complex 3D shapes. 

Moreover, it is interesting to note that we can achieve different combinations of 2D/3D and contracted/deployed shapes in different states. To realize these options, we simply need to change the initial cutting angle for the patterns in both states. For instance, as shown in Fig.~\ref{fig:examples-2D-to-3D}\textbf{b}, we can create kirigami structures from a 2D compact rectangle to a 3D deployed shape with negative space approximating a half-vase (achieved by setting initial cutting angles $\theta_1=0, \theta_2\in(0, \pi)$), or from a 2D deployed rectangular shape with negative space to a 3D compact shape approximating the half-vase (achieved by setting $0<\theta_1<\pi$ and $\theta_2 = 0$ or $\pi$). We can even achieve a structure that morphs from a 2D compact shape to a 3D reconfigured compact shape approximating the half-vase (achieved by setting $\theta_1=0, \theta_2 = \pi$).

We remark that while there is great flexibility in controlling the shape change achieved by 2D-to-3D kirigami structures, certain theoretical limitations are imposed by the target shapes and the desired kirigami properties. For instance, one can prove mathematically that it is impossible to construct a kirigami structure that can morph from a compact 2D rectangle to a positively curved 3D compact surface, such as a spherical cap, with all vertices lying on the target geometry; for a target of mixed curvature such as the vase, the same argument puts a strong constraint on such a pattern. We leave the theoretical analysis to Section~\ref{sec:limitation}.

\subsection{3D-to-3D kirigami structures}

For the cube-to-sphere structure of Fig.~\ref{fig:examples-2D-to-3D}\textbf{c}, we exploit the symmetries of both shapes and solve the design problem on a small region of the cube and the sphere (SI Appendix~\ref{appendix:additionalresults}). The structure of Fig.~\ref{fig:examples-2D-to-3D}\textbf{d}, which morphs from a square ring to a torus, has states of genus one. A portion of the square ring unfolds isometrically onto the plane, so we solve a 2D-to-3D problem from this planar layout to a portion of the torus and assemble copies of the result by symmetry. Alternatively, both states can be optimized directly in 3D: SI Fig.~\ref{fig:SI_results_3Dto3D} shows compact reconfigurable structures whose first compact state approximates a spherical cap and whose second compact state approximates a saddle and other shapes.

\subsection{Multi-state structures}
Finally, we design one flat sheet that deploys onto two prescribed 3D shapes. Figure~\ref{fig:multistate} shows a pattern with $M=N=10$, a free compact outline and tile angles kept in $[\pi/9,8\pi/9]$, whose two deployed states match a bowl $z=0.4(x^2+y^2)$ and a saddle $z=0.4(x^2-y^2)$ over the unit disk. Both states are initialized at $\theta=\pi/2$ and optimized together for 500 iterations. The root-mean-square distances of the deployed vertices to the two targets are $0.0035$ and $0.0155$ mean tile edges. The two states share every tile and differ only in how far each cell opens: measured in each group of $2\times2$ neighboring tiles (SI Appendix~\ref{appendix:tess}), all openings $\phi=\theta/2$ of the bowl lie below $\pi/4$, where a pattern of square tiles is fully open, and those of $75$ of the $81$ groups of the saddle above it (Fig.~\ref{fig:multistate}\textbf{a},\textbf{b}).

The two states are connected by a continuous deformation of the same sheet. We hold the compact pattern fixed and deploy it onto the intermediate targets $z=0.4(x^2+t\,y^2)$, from $t=1$ (bowl) through $t=0$ (cylinder) to $t=-1$ (saddle) in steps of $0.05$, each step starting from the previous one. Every intermediate state is an exact configuration of the rigid tiles, and the path is the same in both directions (Fig.~\ref{fig:multistate}\textbf{c}); the intermediate targets themselves are matched only approximately, with root-mean-square distances up to $0.145$ mean tile edges near the cylinder. Along the path, the median opening increases steadily and crosses $\pi/4$ at $t\approx0.15$. The change of the opening of each block follows the Gaussian curvature $K$ of the target at that block: $\phi-\phi_{t=0}\approx-0.66\,K$ ($K$ for the unit disk) over all blocks and all $t$, with correlation $-0.99$, flattening at the largest $|K|$ (Fig.~\ref{fig:multistate}\textbf{d}). More analysis can be seen in Section~\ref{sec:curvature}.

\begin{figure*}[t]
    \centering
    \includegraphics[width=\linewidth]{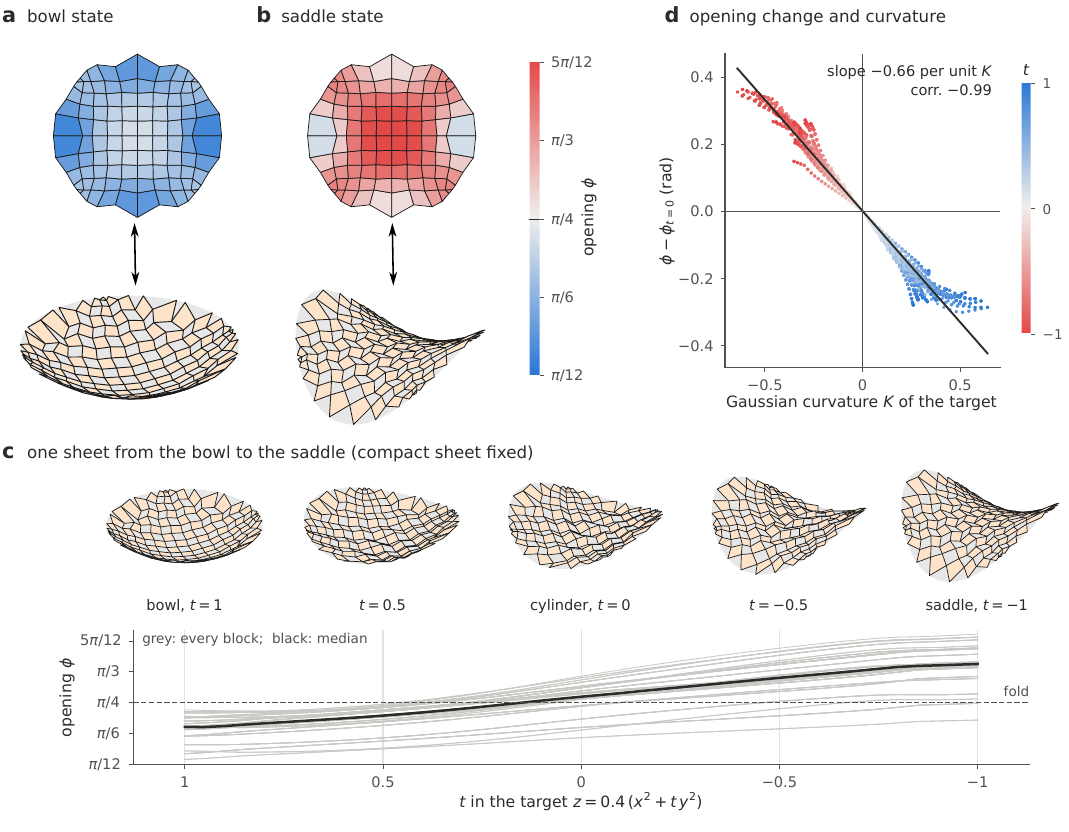}
    \caption{\textbf{Multi-state kirigami structures designed by our framework.} Here, a free compact pattern ($M=10$) is designed to deploy into both a bowl $z=0.4(x^2+y^2)$ and a saddle $z=0.4(x^2-y^2)$ over the unit disk (targets in grey). \textbf{a}, The compact sheet colored by the opening $\phi$ of the bowl state, on a color scale centered at $\phi=\pi/4$, and the deployed bowl. The opening is half the relative rotation of neighboring tiles, averaged over each group of $2\times2$ tiles (SI Appendix~\ref{appendix:tess}); each tile shows the mean opening of the groups it belongs to. \textbf{b}, The same for the saddle state. \textbf{c}, Continuous deformation of the same sheet from the bowl to the saddle: deployed states for the targets $z=0.4(x^2+t\,y^2)$ at $t=1$ (bowl), $0.5$, $0$ (cylinder), $-0.5$ and $-1$ (saddle), each above its value of $t$, and the opening of every group of tiles (grey) and their median (black) against $t$, with the fold $\phi=\pi/4$ (dashed). \textbf{d}, Change of the opening of each group of tiles from its value at $t=0$ against the Gaussian curvature $K$ of the target at the group (targets on the unit disk), for all $81$ groups and all $41$ values of $t$, colored by $t$ (opening in radians); line: least-squares fit.}
    \label{fig:multistate}
\end{figure*}

\section{Theoretical analysis on kirigami design} \label{sec:analysis}
The designs of Section~\ref{sec:results} were obtained by optimization alone, yet they share regularities that no constraint of the optimization imposes, and they run into limits that no choice of its parameters removes. Three of them are the subject of this section. First, compact reconfigurable patterns appear to turn the shape they fill by a quarter turn: a wide target is reached from a tall compact rectangle, a tall target from a wide one, and the second moments of area of the two compact states are almost exactly reciprocal (Section~\ref{sec:harmonic}). Second, a single sheet carries both a bowl and a saddle, and along the deformation from one to the other each group of tiles opens or closes according to the sign of the Gaussian curvature of the target beneath it (Fig.~\ref{fig:multistate}\textbf{d}); we explain where the curvature of a deployed kirigami sheet comes from (Section~\ref{sec:curvature}). Third, compact reconfiguration onto curved surfaces is restricted: a compact rectangle cannot be reconfigured onto a sphere (Section~\ref{sec:limitation}). Throughout, we describe a pattern by the lattice of its blocks of $2\times2$ neighboring tiles, labelled by continuous coordinates $(s,t)$, one unit per tile step, with $X(s,t)$ the position of the block centers in the compact state; derivatives of $X$ are central differences of neighboring block centers (SI Fig.~\ref{fig:block_centers}).

\begin{figure*}[t]
    \centering
    \includegraphics[width=\linewidth]{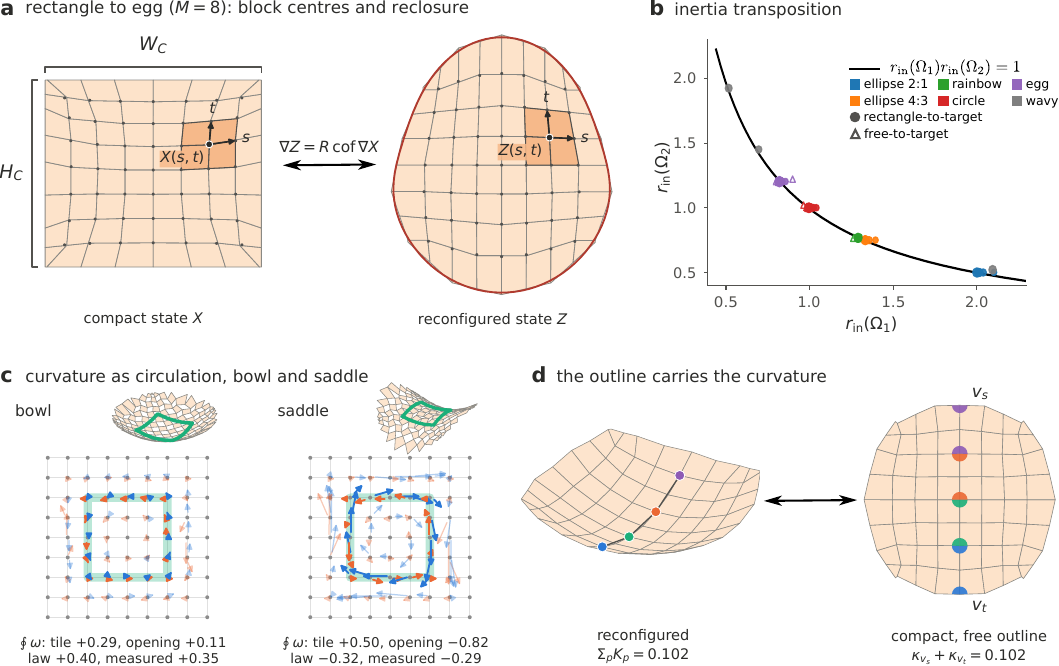}
    \caption{\textbf{Compact reconfiguration and curvature.} \textbf{a}, Rectangle-to-egg compact reconfigurable design ($M=8$): the compact $W_C\times H_C$ rectangle (left) and the reconfigured compact state on the egg target (right, red), related by Eq.~\eqref{eq:reclosure-law}. Dots: the block centers $X(s,t)$ and $Z(s,t)$, each the mean of the four tile centers of a block of $2\times2$ tiles, labelled by the block coordinates $(s,t)$; one block is shaded in both states, with the directions of $s$ and $t$ (SI Fig.~\ref{fig:block_centers}). \textbf{b}, Inertia transposition: the inertia ratio $r_{\mathrm{in}}=\sqrt{I_{xx}/I_{yy}}$ of the reconfigured compact state $\Omega_2$ against that of the first compact state $\Omega_1$, for the $42$ compact reconfigurable designs with different target shapes and resolutions. \textbf{c}, Curvature as circulation (Eq.~\eqref{eq:gauss}) on the bowl and the saddle of Fig.~\ref{fig:multistate}, which share one compact pattern. Top: the deployed state and the loop of blocks one ring inside the boundary (green). Bottom: the lattice of blocks in the coordinates $(s,t)$, with the connection form $\boldsymbol\omega$ at every interior block, split into its tile part (orange, from $\Delta X$) and its opening part (blue, from $\nabla\phi$); arrows off the loop are pale. Below: the circulation of each part around the loop, their sum (law) and the curvature inside the loop measured from the angle defects. \textbf{d}, The compact outline carries the curvature of the reconfigured compact state (Eq.~\eqref{eq:chain}). A compact reconfigurable 2D-to-3D design ($M=8$): its reconfigured state on a bowl (left) and its compact state, with a free outline (right); the angle defects at the four interior vertices $p$ of one hinge chain (discs, one color per vertex) are carried in the compact state by the same tile corners (wedges), and their sum equals the turning of the compact outline at the two ends $v_s$, $v_t$ of the chain. Angles in radians.}
    \label{fig:demo_analysis}
\end{figure*}

\subsection{Compact reconfiguration and the inertia transposition}\label{sec:harmonic}

As the opening of a uniform pattern grows from $0$ to $\pi/2$, every tile turns by $\pm\pi/2$ in a checkerboard pattern, and the pattern closes again; for rectangular tiles the mechanism behaves like a coordinate \emph{transposition}, exchanging the width and the height of the enclosing rectangle (Fig.~\ref{fig:tess_terminology}\textbf{a}). We ask whether this transposition survives in optimized compact reconfigurable patterns, whose tiles are no longer congruent rectangles, and what it implies about the shapes that one pattern can connect. In an optimized pattern, each block still deploys nearly as a uniform pattern does, with the cell law of Zheng et al.~\cite{zheng2022continuum} applied block by block (SI Appendix~\ref{appendix:tess}). Our planar designs are rigidly reconfigurable, and every hinge opens by the same angle along their mechanism (Section~\ref{sec:reconfig-rigid}), so at reclosure every tile has turned by $\pm\pi/2$ relative to one common rotation, and the block centers $Z$ of the reconfigured compact state are tied to the block centers $X$ of the first compact state by one rotation $R$ (Fig.~\ref{fig:demo_analysis}\textbf{a}):
\begin{equation}
\label{eq:reclosure-law}
\nabla Z=R\operatorname{cof}\nabla X,
\end{equation}
up to a discretization error that decreases as $(MN)^{-1}$ for an $M \times N$ pattern (SI Appendix~\ref{appendix:lattice}).

The transposition survives in the lattice. Written out, for $Z = (Z_1, Z_2)$ and $X = (X_1, X_2)$, we have
\begin{equation*}
\operatorname{cof}\nabla X=\begin{pmatrix}\partial_tX_2&-\partial_sX_2\\-\partial_tX_1&\partial_sX_1\end{pmatrix},
\end{equation*}
and with $R=I$ the diagonal entries of Eq.~\eqref{eq:reclosure-law} read $\partial_tZ_2=\partial_sX_1$ and $\partial_sZ_1=\partial_tX_2$: the vertical spacing of consecutive rows in the reconfigured state equals the horizontal spacing of consecutive columns in the compact state, and vice versa. A compact $M\times N$ rectangle of width $W_C$ and height $H_C$ therefore reconfigures into a state whose rows are spaced by $W_C/M$ on average and whose columns by $H_C/N$, whatever the target; with the vertex means of the tiles as centers, this holds exactly at every resolution (SI Appendix~\ref{appendix:transpose}).

The equation also constrains the compact lattice itself. Since $R$ is one rotation, the rows of $\operatorname{cof}\nabla X$ are the gradients of the two components of $R^{\mathsf T}Z$, so their curls, $-\Delta X_2$ and $\Delta X_1$, vanish: $\Delta X=0$, and the compact lattice is harmonic. A harmonic lattice is fixed by its boundary values, so the outline decides the whole pattern; designing the outline alone, with no edge-length equalities and none of their redundancies, is much faster than the full optimization (SI Appendix~\ref{appendix:boundary-design}).

The rows and columns are transposed, but is the shape they fill transposed as well? We measure the shape of a region $\Omega$ by its central second moments of area, $I_{xx}(\Omega)=\iint_\Omega(y-\bar y)^2\,dA$ and $I_{yy}(\Omega)=\iint_\Omega(x-\bar x)^2\,dA$, and by the inertia ratio
\begin{equation}
\label{eq:rin-def}
r_{\mathrm{in}}(\Omega)=\sqrt{\frac{I_{xx}(\Omega)}{I_{yy}(\Omega)}},
\end{equation}
which is $H_C/W_C$ for a $W_C\times H_C$ rectangle and $b/a$ for an ellipse with semi-axes $a$ along $x$ and $b$ along $y$. Figure~\ref{fig:demo_analysis}\textbf{b} compares the inertia ratios of the two compact states $\Omega_1$ and $\Omega_2$ for all the planar compact reconfigurable designs of this work with $M=N$, rectangle-to-target and free-to-target alike (SI Fig.~\ref{fig:targets}). They lie on the hyperbola $r_{\mathrm{in}}(\Omega_1)\,r_{\mathrm{in}}(\Omega_2)=1$, to within $2\%$ from $M=12$ on, except for one degenerate design: the reconfigured state has the moment aspect of the compact state, turned by a quarter turn. The mechanism behind this is simple. A quarter turn exchanges the two second moments of every tile about its own centroid, and as the tiles shrink, the moments of a whole state are carried by the spread of the tile centroids; provided this spread is balanced between the two states, the transposition passes from the tiles to the shapes.

\begin{theorem}[Inertia transposition]
\label{thm:inertia}
Let compact reconfigurable patterns with $M=N$ have planar compact states that converge, as $M\to\infty$ and the tiles shrink to points, to regions $\Omega_1$ and $\Omega_2$. If the tile centroids of the two states spread in balance (condition (H3) of SI Theorem~\ref{thm:inertia-free}), then
\begin{equation}
\label{eq:rin-reciprocal}
r_{\mathrm{in}}(\Omega_1)\,r_{\mathrm{in}}(\Omega_2)=1 .
\end{equation}
\end{theorem}

The proof is in SI Appendix~\ref{appendix:inertia}. The balance is not implied by the mechanism alone: an hourglass, or a rounded triangle whose second moments are isotropic, is reached from a compact rectangle whose aspect misses Eq.~\eqref{eq:rin-reciprocal}, if only by a few percent (SI Fig.~\ref{fig:SI_hourglass}). For a compact rectangle, the theorem becomes a design rule.

\begin{cor}[Aspect-ratio law]
\label{cor:aspect}
A target $\Omega$ of area $\mathcal A$ that satisfies Theorem~\ref{thm:inertia} is reached from a compact rectangle of width and height
\begin{equation}
\label{eq:inertia-law}
W_C=\sqrt{\mathcal A\,r_{\mathrm{in}}(\Omega)},\qquad H_C=\sqrt{\mathcal A/r_{\mathrm{in}}(\Omega)} .
\end{equation}
\end{cor}

A tall target thus needs a wide compact rectangle, and a wide target a tall one. In the rectangle-to-target designs, the compact aspect $W_C/H_C$ converges to $r_{\mathrm{in}}(\Omega)$ to within $0.6\%$ at $M=24$, so Eq.~\eqref{eq:inertia-law} can be used to choose the compact rectangle before the optimization.

\subsection{Curvature of 2D-to-3D designs}\label{sec:curvature}

A 2D-to-3D design must produce Gaussian curvature from rigid flat tiles, and Fig.~\ref{fig:multistate} shows that one compact pattern can produce curvature of either sign. Where does the curvature come from? By Gauss's Theorema Egregium, it is fixed by the metric of the deployed sheet, and the metric of a block is fixed by its compact shape and its opening. With the opening $\phi(s,t)$ and a rotation $\gamma(s,t)$ that vary across the sheet, the cell law applied block by block gives the deployed lattice Jacobian, in the tangent plane of the deployed sheet,
\begin{equation}
\label{eq:lattice-law}
R(\gamma)\,\Phi,\qquad \Phi=\cos\phi\,\nabla X+\sin\phi\operatorname{cof}\nabla X ,
\end{equation}
exactly for a uniform pattern and up to second differences of the tile shapes for an optimized one (SI Appendix~\ref{appendix:tess}); at $\phi=\pi/2$ it is the reclosure law~\eqref{eq:reclosure-law}. The rows of $\Phi$ are an orthonormal coframe of the deployed metric $\Phi^{\mathsf T}\Phi$. They fail to be gradients by the torsion $\mathbf q$, and Cartan's structure equations turn this failure into curvature (SI Appendix~\ref{appendix:tess}):
\begin{equation}
\label{eq:gauss}
\begin{gathered}
K\,\det\Phi=\operatorname{curl}\boldsymbol\omega,\qquad
\boldsymbol\omega=-\frac{\Phi^{\mathsf T}\,\mathbf q}{\det\Phi},\\
\mathbf q=\sin\phi\;J\,\Delta X\;+\;\partial_\phi\Phi\;J\,\nabla\phi,
\end{gathered}
\end{equation}
where $K$ is the Gaussian curvature of the deployed state and $J$ the rotation by $\pi/2$. By Stokes' theorem, the total curvature inside any loop of blocks is the circulation of the connection form $\boldsymbol\omega$ around it (Fig.~\ref{fig:demo_analysis}\textbf{c}). Since $\boldsymbol\omega$ is linear in $\mathbf q$, this curvature splits exactly into two parts with distinct origins: a \emph{tile part}, from the Laplacian of the compact lattice, i.e.\ from how the tile shapes vary across the pattern, and an \emph{opening part}, from the gradient of the opening. A uniform opening on a harmonic compact lattice has neither, and stays flat at every opening, as do the planar designs of Section~\ref{sec:harmonic}. That a kirigami surface is curved according to its deployed metric is classical, and has been used for conformal auxetic patterns~\cite{konakovic2016beyond}; to our knowledge, the split of Eq.~\eqref{eq:gauss} into a tile part and an opening part has not been given before.

The opening part has a sign that depends on the opening itself. Its weight is $\partial_\phi\Phi=-\sin\phi\,\nabla X+\cos\phi\operatorname{cof}\nabla X$; splitting $\nabla X$ into its conformal part $\mathcal C$ and its anticonformal part $\mathcal N$ gives $\partial_\phi\Phi=(\cos\phi-\sin\phi)\,\mathcal C-(\cos\phi+\sin\phi)\,\mathcal N$, whose conformal weight changes sign at the fold $\phi=\pi/4$, where a block of any shape has its largest area. The same opening gradient therefore curves a state below the fold one way and a state above it the other way.

This is how the sheet of Fig.~\ref{fig:multistate} carries both a bowl and a saddle (Fig.~\ref{fig:demo_analysis}\textbf{c}). The two states share their tiles, and the tile part, positive in both, is the same up to the change of the opening. In the bowl, whose openings lie below the fold, the opening part adds to it; in the saddle, whose openings lie above, it turns around and outweighs it. Around the loop, one ring inside the boundary, the law gives $+0.40$ for the bowl and $-0.32$ for the saddle, against the measured $+0.35$ and $-0.29$. The same sign change explains why the openings follow the curvature of the target along the path between the two states (Fig.~\ref{fig:multistate}\textbf{d}): the path crosses the fold near the cylinder, so positive curvature is carried by openings below $\pi/4$ and negative curvature by openings above it. Taking the linear form $\phi-\phi_{t=0}=cK$, the law on the interior loops asks for $c\approx-0.54$, against the measured $-0.66$. The law is least accurate near the cylinder, where the path crosses the fold: there the metric of a block is stationary in $\phi$, and the rim of the sheet, held on each intermediate target, distorts the blocks to make the crossing smooth; with a ten times weaker rim the law holds along the whole path, but the openings then change fastest at the fold.

\subsection{Reclosure on curved surfaces}\label{sec:limitation}

The 2D-to-3D compact reconfigurable designs of Section~\ref{sec:results} reach their curved compact states from free compact outlines (Fig.~\ref{fig:examples-2D-to-3D}\textbf{b}, Fig.~\ref{fig:demo_analysis}\textbf{d}). Most of our planar designs, by contrast, start from a compact rectangle. Can a compact rectangle also be reconfigured onto a curved surface, and why should the outline of the planar state matter for the curvature of the other compact state?

A special case of Eq.~\eqref{eq:gauss} answers this. In the reconfigured compact state $\phi=\pi/2$ in every block, so $\Phi=\operatorname{cof}\nabla X$, $\mathbf q=J\,\Delta X$, and the opening part disappears:
\begin{equation}
\label{eq:reclosed-gauss}
K\det\nabla X=-\operatorname{div}\!\left(\frac{\nabla X^{\mathsf T}\Delta X}{\det\nabla X}\right).
\end{equation}
The curvature of the reconfigured state is the divergence of a field of the planar compact lattice, so the curvature inside a region is the flux of that field through its boundary, and over the whole sheet it is determined at the compact outline.

On the pattern this statement holds exactly, with angles in place of derivatives. At a hinge $e$, where the corners of two tiles $A$ and $B$ coincide in every state, let $\xi_e=\alpha_A+\alpha_B-\pi$ be the bend of the slit, with $\alpha_A,\alpha_B$ the two tile angles (SI Fig.~\ref{fig:SI_chains}\textbf{a}). The slits of a rigidly reconfigurable pattern are straight, $\xi_e=0$, the discrete counterpart of $\Delta X=0$; those of a 2D-to-3D design bend. The four tile angles at an interior vertex of the planar compact state close to $2\pi$, and those at an interior vertex $p$ of the reconfigured state to $2\pi-K_p$, with $K_p$ the angle defect. Linking hinges alternately through the vertices of the two states arranges them into chains, half of which begin and end on the compact outline (SI Fig.~\ref{fig:SI_chains}\textbf{c}). Along such a chain the bends cancel in pairs, as the interior fluxes do in a divergence theorem, and only its two ends $v_s,v_t$ remain:
\begin{equation}
\label{eq:chain}
\sum_{p\in\text{chain}}K_p=\kappa_{v_s}+\kappa_{v_t},
\end{equation}
where $\kappa_v$ is the turning of the compact outline at an end (SI Appendix~\ref{appendix:chains}). In the bowl of Fig.~\ref{fig:demo_analysis}\textbf{d}, the angle defects at the four vertices of one chain (discs) are carried in the planar state by the same tile corners (wedges); these close up at every interior vertex, and their sum, $0.102$, reappears as the turning of the outline at the two ends of the chain.

The ends of these chains lie on the sides of the compact outline, not at its corners, and the sides of a rectangle do not turn. A compact rectangle therefore forces the angle defects along every such chain to sum to zero. On a positively curved surface such as the bowl or a sphere, this cannot hold with all vertices on the surface. On a target of mixed curvature, such as the vase of Section~\ref{sec:results}, a compact rectangle would need every such chain to balance its positive and negative curvature exactly. For the same reason, the assembly approach of Fig.~\ref{fig:examples-2D-to-3D}\textbf{c} cannot produce a compact-cube-to-compact-sphere structure, which would need a square-to-spherical-cap compact reconfigurable design for each face of the cube. Conversely, a rigidly reconfigurable pattern, whose slits are all straight, has $K_p=0$ at every vertex and reconfigures only into a developable surface.

\section{Discussion}

Inverse kirigami design has so far been posed one kind of shape change at a time. We have shown that these problems share one structure: the states of a pattern differ only in which tile corners coincide, one routine generates them, and edge pairs couple them. A second compact state, a 3D target or a second prescribed shape, then becomes one more state or constraint, solved by the same code, and the redundancies among the lengths need to be analyzed only once (SI Appendix~\ref{appendix:optimization}).

Comparing designs across problems also revealed regularities that the analysis of Section~\ref{sec:analysis} explains, and the explanations feed back into the design. The aspect-ratio law~\eqref{eq:inertia-law} chooses the compact rectangle before any optimization; harmonicity reduces the design of planar compact reconfigurable patterns to their outline, which is much faster than the full optimization (SI Appendix~\ref{appendix:boundary-design}); the fold tells which openings carry positive and which negative curvature; and the chain identity~\eqref{eq:chain} tells which compact outlines can reach a curved reconfigured state. Two questions remain open: which targets satisfy the balance condition of the inertia transposition (SI Fig.~\ref{fig:SI_hourglass}), and whether a compact rectangle can be reconfigured onto a target of mixed curvature, where the chain identity constrains the design without deciding it.

Our framework assumes rigid planar quadrilateral tiles with straight edges, connected at ideal hinges, whereas real sheets have thickness and elasticity, and our 2D-to-3D designs prescribe the shapes, not how to actuate the sheet into them. Curved cuts and tiles, other tessellations, and actuation are natural extensions. More broadly, as our framework rests on length equalities between states, which are common to many design problems, it may extend to a wide range of mechanical metamaterials and engineering structures beyond kirigami.

\textbf{Acknowledgment} \ We thank the CUHK Library Learning Garden Team for support in physical model fabrication.\\

\bibliographystyle{ieeetr}
\bibliography{reference}

@article{choi2019programming,
  title={Programming shape using kirigami tessellations},
  author={Choi, Gary P. T. and Dudte, Levi H. and Mahadevan, L.},
  journal={Nat. Mater.},
  volume={18},
  number={9},
  pages={999--1004},
  year={2019}
}

@article{choi2021compact,
  title={Compact reconfigurable kirigami},
  author={Choi, Gary P. T. and Dudte, Levi H. and Mahadevan, L.},
  journal={Phys. Rev. Research},
  volume={3},
  number={4},
  pages={043030},
  year={2021},
  publisher={APS}
}

@article{jiang2022shape,
  title={Shape-morphing mechanical metamaterials},
  author={Jiang, Caigui and Rist, Florian and Wang, Hui and Wallner, Johannes and Pottmann, Helmut},
  journal={Comput. Aided Des.},
  volume={143},
  pages={103146},
  year={2022},
  publisher={Elsevier}
}

@article{dang2022theorem,
  title={Theorem on the compatibility of spherical kirigami tessellations},
  author={Dang, Xiangxin and Feng, Fan and Duan, Huiling and Wang, Jianxiang},
  journal={Phys. Rev. Lett.},
  volume={128},
  number={3},
  pages={035501},
  year={2022},
  publisher={APS}
}

@article{dang2021theorem,
  title={Theorem for the design of deployable kirigami tessellations with different topologies},
  author={Dang, Xiangxin and Feng, Fan and Duan, Huiling and Wang, Jianxiang},
  journal={Phys. Rev. E},
  volume={104},
  number={5},
  pages={055006},
  year={2021},
  publisher={APS}
}

@article{jiang2022flexible,
  title={Flexible metamaterial electronics},
  author={Jiang, Shan and Liu, Xuejun and Liu, Jianpeng and Ye, Dong and Duan, Yongqing and Li, Kan and Yin, Zhouping and Huang, YongAn},
  journal={Adv. Mater.},
  volume={34},
  number={52},
  pages={2200070},
  year={2022},
  publisher={Wiley Online Library}
}

@article{yang2021grasping,
  title={Grasping with kirigami shells},
  author={Yang, Yi and Vella, Katherine and Holmes, Douglas P},
  journal={Sci. Robot.},
  volume={6},
  number={54},
  pages={eabd6426},
  year={2021},
  publisher={American Association for the Advancement of Science}
}

@article{cheng2023programming,
  title={Programming {3D} curved mesosurfaces using microlattice designs},
  author={Cheng, Xu and Fan, Zhichao and Yao, Shenglian and Jin, Tianqi and Lv, Zengyao and Lan, Yu and Bo, Renheng and Chen, Yitong and Zhang, Fan and Shen, Zhangming and others},
  journal={Science},
  volume={379},
  number={6638},
  pages={1225--1232},
  year={2023},
  publisher={American Association for the Advancement of Science}
}

@article{konakovic2016beyond,
  title={Beyond developable: computational design and fabrication with auxetic materials},
  author={Konakovi{\'c}, Mina and Crane, Keenan and Deng, Bailin and Bouaziz, Sofien and Piker, Daniel and Pauly, Mark},
  journal={ACM Trans. Graph.},
  volume={35},
  number={4},
  pages={1--11},
  year={2016},
  publisher={ACM New York, NY, USA}
}

@article{panetta2021computational,
  title={Computational inverse design of surface-based inflatables},
  author={Panetta, Julian and Isvoranu, Florin and Chen, Tian and Si{\'e}fert, Emmanuel and Roman, Beno{\^\i}t and Pauly, Mark},
  journal={ACM Trans. Graph.},
  volume={40},
  number={4},
  pages={1--14},
  year={2021},
  publisher={ACM New York, NY, USA}
}

@article{konakovic2018rapid,
  title={Rapid deployment of curved surfaces via programmable auxetics},
  author={Konakovi{\'c}-Lukovi{\'c}, Mina and Panetta, Julian and Crane, Keenan and Pauly, Mark},
  journal={ACM Trans. Graph.},
  volume={37},
  number={4},
  pages={1--13},
  year={2018},
  publisher={ACM New York, NY, USA}
}

@article{chen2021bistable,
  title={Bistable auxetic surface structures},
  author={Chen, Tian and Panetta, Julian and Schnaubelt, Max and Pauly, Mark},
  journal={ACM Trans. Graph.},
  volume={40},
  number={4},
  pages={1--9},
  year={2021},
  publisher={ACM New York, NY, USA}
}

@article{nojoomi20212d,
  title={{2D} material programming for {3D} shaping},
  author={Nojoomi, Amirali and Jeon, Junha and Yum, Kyungsuk},
  journal={Nat. Commun.},
  volume={12},
  number={1},
  pages={603},
  year={2021},
  publisher={Nature Publishing Group UK London}
}

@article{dudte2023additive,
  title={An additive framework for kirigami design},
  author={Dudte, Levi H. and Choi, Gary P. T. and Becker, Kaitlyn P. and Mahadevan, L.},
  journal={Nat. Comput. Sci.},
  volume={3},
  number={5},
  pages={443--454},
  year={2023},
  publisher={Nature Publishing Group US New York}
}

@article{liu2022quasicrystal,
  title={Quasicrystal kirigami},
  author={Liu, Lucy and Choi, Gary P. T. and Mahadevan, L},
  journal={Phys. Rev. Research},
  volume={4},
  number={3},
  pages={033114},
  year={2022},
  publisher={APS}
}

@article{zheng2022continuum,
  title={Continuum field theory for the deformations of planar kirigami},
  author={Zheng, Yue and Niloy, Imtiar and Celli, Paolo and Tobasco, Ian and Plucinsky, Paul},
  journal={Phys. Rev. Lett.},
  volume={128},
  number={20},
  pages={208003},
  year={2022},
  publisher={APS}
}

@article{bertoldi2017flexible,
  title={Flexible mechanical metamaterials},
  author={Bertoldi, Katia and Vitelli, Vincenzo and Christensen, Johan and Van Hecke, Martin},
  journal={Nat. Rev. Mater.},
  volume={2},
  number={11},
  pages={1--11},
  year={2017},
  publisher={Nature Publishing Group}
}

@article{barchiesi2019mechanical,
  title={Mechanical metamaterials: a state of the art},
  author={Barchiesi, Emilio and Spagnuolo, Mario and Placidi, Luca},
  journal={Math. Mech. Solids},
  volume={24},
  number={1},
  pages={212--234},
  year={2019},
  publisher={SAGE Publications Sage UK: London, England}
}

@article{xia2022responsive,
  title={Responsive materials architected in space and time},
  author={Xia, Xiaoxing and Spadaccini, Christopher M and Greer, Julia R},
  journal={Nat. Rev. Mater.},
  volume={7},
  number={9},
  pages={683--701},
  year={2022},
  publisher={Nature Publishing Group UK London}
}

@article{liu2021wallpaper,
  title={Wallpaper group kirigami},
  author={Liu, Lucy and Choi, Gary P. T. and Mahadevan, L},
  journal={Proc. R. Soc. A},
  volume={477},
  number={2252},
  pages={20210161},
  year={2021},
  publisher={The Royal Society Publishing}
}

@article{li2023auxetic,
  title={Auxetic mechanical metamaterials: from soft to stiff},
  author={Li, Xiang and Peng, Weitao and Wu, Wenwang and Xiong, Jian and Lu, Yang},
  journal={Int. J. Extreme Manuf.},
  volume={5},
  number={4},
  pages={042003},
  year={2023},
  publisher={IOP Publishing}
}

@article{shim2012buckling,
  title={Buckling-induced encapsulation of structured elastic shells under pressure},
  author={Shim, Jongmin and Perdigou, Claude and Chen, Elizabeth R and Bertoldi, Katia and Reis, Pedro M},
  journal={Proc. Natl. Acad. Sci.},
  volume={109},
  number={16},
  pages={5978--5983},
  year={2012},
  publisher={National Acad Sciences}
}

@article{haghpanah2016multistable,
  title={Multistable shape-reconfigurable architected materials},
  author={Haghpanah, Babak and Salari-Sharif, Ladan and Pourrajab, Peyman and Hopkins, Jonathan and Valdevit, Lorenzo},
  journal={Adv. Mater.},
  volume={28},
  number={36},
  pages={7915--7920},
  year={2016}
}

@article{overvelde2017rational,
  title={Rational design of reconfigurable prismatic architected materials},
  author={Overvelde, Johannes T. B. and Weaver, James C and Hoberman, Chuck and Bertoldi, Katia},
  journal={Nature},
  volume={541},
  number={7637},
  pages={347--352},
  year={2017},
  publisher={Nature Publishing Group UK London}
}

@article{li20243d,
  title={{3D} Tiled Auxetic Metamaterial: A New Family of Mechanical Metamaterial with High Resilience and Mechanical Hysteresis},
  author={Li, Tiantian and Li, Yaning},
  journal={Adv. Mater.},
  volume={36},
  number={15},
  pages={2309604},
  year={2024},
  publisher={Wiley Online Library}
}

@article{toyonaga2026collapsible,
  title={Collapsible scissored surfaces},
  author={Toyonaga, Noah and Nishimoto, Seri and Decker, Colter J. and Wood, Robert J. and Tachi, Tomohiro and Mahadevan, L.},
  journal={Proc. Natl. Acad. Sci.},
  volume={123},
  number={25},
  pages={e2605221123},
  year={2026},
}

@article{gu2024kirigami,
  title={Kirigami-Inspired Three-Dimensional Metamaterials with Programmable Isotropic and Orthotropic Thermal Expansion},
  author={Gu, Yuanqing and Wei, Zhibo and Wei, Guowu and You, Zhong and Ma, Jiayao and Chen, Yan},
  journal={Adv. Mater.},
  pages={2411232},
  year={2024},
  publisher={Wiley Online Library}
}

@article{dang2025shape,
  title={Shape and topology morphing of closed surfaces integrating origami and kirigami},
  author={Dang, Xiangxin and Chen, Shujia and Acha, Ali Elias and Wu, Lei and Pasini, Damiano},
  journal={Sci. Adv.},
  volume={11},
  number={18},
  pages={eads5659},
  year={2025},
  publisher={American Association for the Advancement of Science}
}

@article{qiao2025inverse,
  title={Inverse Design of Kirigami through Shape Programming of Rotating Units},
  author={Qiao, Chuan and Chen, Shijun and Chen, Yu and Zhou, Zhihong and Jiang, Wentao and Wang, Qingyuan and Tian, Xiaobao and Pasini, Damiano},
  journal={Phys. Rev. Lett.},
  volume={134},
  number={17},
  pages={176103},
  year={2025},
  publisher={APS}
}

@article{zhai2021mechanical,
  title={Mechanical metamaterials based on origami and kirigami},
  author={Zhai, Zirui and Wu, Lingling and Jiang, Hanqing},
  journal={Appl. Phys. Rev.},
  volume={8},
  number={4},
  year={2021},
  publisher={AIP Publishing}
}

@article{jin2024engineering,
  title={Engineering kirigami frameworks toward real-world applications},
  author={Jin, Lishuai and Yang, Shu},
  journal={Adv. Mater.},
  volume={36},
  number={9},
  pages={2308560},
  year={2024},
  publisher={Wiley Online Library}
}

@article{jiang2025pykirigami,
  title={{PyKirigami}: An interactive {P}ython simulator for kirigami structures},
  author={Jiang, Qinghai and Choi, Gary P. T.},
  journal={Comput. Phys. Commun.},
  volume={328},
  pages={110349},
  year={2026}
}

@article{lamoureux2025kirigami,
  title={Kirigami-inspired parachutes with programmable reconfiguration},
  author={Lamoureux, Danick and Fillion, J{\'e}r{\'e}mi and Ramananarivo, Sophie and Gosselin, Fr{\'e}d{\'e}rick P and Melancon, David},
  journal={Nature},
  volume={646},
  number={8083},
  pages={88--94},
  year={2025},
  publisher={Nature Publishing Group UK London}
}

@article{hong2025reprogrammable,
  title={Reprogrammable snapping morphogenesis in ribbon-cluster meta-units using stored elastic energy},
  author={Hong, Yaoye and Zhou, Caizhi and Qing, Haitao and Chi, Yinding and Yin, Jie},
  journal={Nat. Mater.},
  volume={24},
  number={11},
  pages={1793--1801},
  year={2025},
  publisher={Nature Publishing Group UK London}
}

@article{zaman2025one,
  title={One String to Pull Them All: Fast Assembly of Curved Structures from Flat Auxetic Linkages},
  author={Zaman, Akib and Aslarus, Jacqueline and Li, Jiaji and Mueller, Stefanie and Konakovic Lukovic, Mina},
  journal={ACM Trans. Graph.},
  volume={44},
  number={6},
  pages={1--18},
  year={2025},
  publisher={ACM New York, NY, USA}
}

@inproceedings{segall2025reconfigurable,
  title={Reconfigurable Hinged Kirigami Tessellations},
  author={Segall, Aviv and Ren, Jing and Padilla, Marcel and Sorkine-Hornung, Olga},
  booktitle={Proceedings of the SIGGRAPH Asia 2025 Conference Papers},
  number = {99},
  pages={1--11},
  year={2025}
}

@article{giomi2026stretching,
  title={Stretching theory of {H}ookean metashells},
  author={Giomi, Luca},
  journal={Phys. Rev. X},
  volume={16},
  number={2},
  pages={021035},
  year={2026},
  publisher={APS}
}

@article{segall2026uniformly,
  title={Uniformly Deployable Kirigami on Arbitrary Planar Graphs},
  author={Segall, Aviv and Ren, Jing and Sorkine-Hornung, Olga},
  journal={ACM Trans. Graph.},
  volume={45},
  number={4},
  pages={1--17},
  year={2026},
  publisher={ACM New York, NY, USA}
}

@article{wachter2006implementation,
  title={On the implementation of an interior-point filter line-search algorithm for large-scale nonlinear programming},
  author={W{\"a}chter, Andreas and Biegler, Lorenz T},
  journal={Math. Program.},
  volume={106},
  number={1},
  pages={25--57},
  year={2006},
  publisher={Springer}
}

\clearpage

\centerline{\large\textbf{Supplementary Information}}
\appendix
\renewcommand{\topfraction}{0.9}\renewcommand{\dbltopfraction}{0.9}\renewcommand{\textfraction}{0.08}
\renewcommand{\floatpagefraction}{0.8}\renewcommand{\dblfloatpagefraction}{0.8}\setcounter{dbltopnumber}{2}
\renewcommand\thefigure{S\arabic{figure}}    
\setcounter{figure}{0}
\renewcommand\thetable{S\arabic{table}}    
\setcounter{table}{0}


\section{The optimization} \label{appendix:optimization}

\subsection{Residuals and Jacobians} \label{appendix:residualandjacobian}
The variable $x$ stacks the coordinates of all vertices of both states, two per vertex in 2D and three in 3D. Every constraint involves a few vertices only, so its Jacobian row has a few nonzero entries; we assemble the Jacobians of the equality and inequality constraints, $J_e$ and $J_a$, in triplet form for the interior-point solver. There are three kinds of constraints.

\emph{Edge lengths.} An edge $\mathbf p_a\mathbf p_b$ of one state and the corresponding edge $\mathbf p_c\mathbf p_d$ of the other must have equal lengths,
\begin{equation}\label{eq:jac_edge}
r_{\text{edge}}=\|\mathbf p_b-\mathbf p_a\|^2-\|\mathbf p_d-\mathbf p_c\|^2=0,
\end{equation}
with gradient $\mp2(\mathbf p_b-\mathbf p_a)$ with respect to $\mathbf p_a,\mathbf p_b$ and $\pm2(\mathbf p_d-\mathbf p_c)$ with respect to $\mathbf p_c,\mathbf p_d$. Below a length of $10^{-4}$ we use the smoothed difference of lengths $\sqrt{\ell^2+\epsilon^2}$, $\epsilon=10^{-8}$, which keeps the gradient bounded.

\emph{Signed areas.} For three vertices in counter-clockwise order,
\begin{equation}\label{eq:jac_area}
r_{\text{area}}=(x_b-x_a)(y_c-y_a)-(y_b-y_a)(x_c-x_a),
\end{equation}
with gradient $(y_b-y_c,\,x_c-x_b)$, $(y_c-y_a,\,x_a-x_c)$ and $(y_a-y_b,\,x_b-x_a)$ with respect to $\mathbf p_a,\mathbf p_b,\mathbf p_c$. The inequality $r_{\text{area}}\ge\rho>0$ keeps the cutting angles positive and every interior angle of a tile in $(0,\pi)$, and the equality $r_{\text{area}}=0$ makes a slit straight.

\emph{Planarity} (3D states). A tile $\mathbf p_a\mathbf p_b\mathbf p_c\mathbf p_d$ is planar when the volume
\begin{equation}\label{eq:jac_planar}
r_{\text{planar}}=(\mathbf v_{ab}\times\mathbf v_{ac})\cdot\mathbf v_{ad},\qquad \mathbf v_{ak}=\mathbf p_k-\mathbf p_a,
\end{equation}
vanishes; its gradient is $\mathbf v_{ac}\times\mathbf v_{ad}$, $\mathbf v_{ad}\times\mathbf v_{ab}$ and $\mathbf v_{ab}\times\mathbf v_{ac}$ with respect to $\mathbf p_b,\mathbf p_c,\mathbf p_d$, and minus their sum with respect to $\mathbf p_a$.

\subsection{Why one diagonal of each tile is left out} \label{appendix:KKT}
A tile is fixed up to a rigid motion by its four sides and its two diagonals, yet we impose only five of these six lengths. The reason is that the sixth constraint is dependent on the other five near every feasible point, and a dependent row makes the step of the interior-point method,
\begin{equation}
\label{eq:KKT}
\begin{bmatrix} H & J_e^{\mathsf T} & J_a^{\mathsf T}\\ J_e & 0 & 0\\ J_a & 0 & 0\end{bmatrix}
\begin{bmatrix}\Delta x\\ \Delta\lambda\\ \Delta\mu\end{bmatrix}
=-\begin{bmatrix} r_d\\ r_e\\ r_a\end{bmatrix},
\end{equation}
singular. Let the sides of a tile have squared lengths $x,y,z,w$ in cyclic order and let $P$ be the squared length of the diagonal $d_1$; the law of cosines in the two triangles on $d_1$ gives the other diagonal explicitly,
\begin{equation}\label{eq:F_explicit}
d_2^2=F(x,y,z,w,P)=x+w-\frac{AB-\sqrt{UV}}{2P},
\end{equation}
with $A=P+x-y$, $B=P+w-z$, $U=4xP-A^2$ and $V=4wP-B^2$; $U$ and $V$ are $16$ times the squared areas of the two triangles. We impose the five equalities
\begin{equation}\label{eq:five_cons}
e_i=u_i-\tilde u_i=0,\qquad (u_1,\dots,u_5)=(x,y,z,w,P),
\end{equation}
between the two states (tilde). The sixth would be $e_6=F(u)-F(\tilde u)$. At a point where Eq.~\eqref{eq:five_cons} holds, $\partial F/\partial u_i=\partial F/\partial\tilde u_i=:\beta_i$, so
\begin{equation}\label{eq:gradient_dependency}
\nabla e_6=\sum_{i=1}^5\beta_i\,\nabla e_i .
\end{equation}
The coefficients $\beta_i$ depend on the geometry, so the dependency is not a constant linear combination, and a presolver such as the dependency detector of IPOPT, which removes only constant combinations, does not find it. We therefore leave $d_2$ out and let the angle inequalities select the right one of the two quadrilaterals with the five given lengths.

How accurately is the omitted diagonal then matched? If the five constraints hold to a tolerance $\varepsilon$, the mean-value theorem applied to $F$, which is smooth while $P,U,V>0$, gives
\begin{equation}\label{eq:Cq}
|e_6|\le C_q\,\varepsilon,\qquad C_q=\sum_{i=1}^5|\beta_i|,
\end{equation}
and since $|\ell-\tilde\ell|=|\ell^2-\tilde\ell^2|/(\ell+\tilde\ell)$, every length, the omitted diagonal included, is matched to $O(\varepsilon)$ as long as no edge shrinks to zero and no tile becomes a triangle. The angle inequalities keep $U,V\ge4\rho^2$, and a target rescaled so that the initial edges have length of order one (Appendix~\ref{appendix:init}) keeps $C_q$ of order one.

\subsection{Fewer lengths in compact reconfigurable patterns}\label{appendix:reconfig-pairs}
When both states are compact, some edge pairs are implied by others. Around each interior vertex, the four sides of the negative space coincide in pairs one way in the first compact state and the other way in the second (main text Fig.~\ref{fig:tess_terminology}\textbf{c}), so their four length equalities form a cycle whose residuals satisfy $r_1-r_2+r_3-r_4\equiv0$. The $5N^2$ equalities of an $N\times N$ pattern therefore contain $(N-1)^2$ exact dependencies, and for $N\ge9$ they outnumber the $4(N+1)^2$ variables; dropping one side of every interior tile keeps their rank and leaves $4N^2+4N-4$ of them.

A compact rectangle adds $N-1$ dependencies of another kind, one per $C$-chain of Appendix~\ref{appendix:chains}, which hold only at feasible points: once the tiles are congruent and the second state is planar, Eq.~\eqref{eq:chain} makes the turnings of the outline at the two ends of a $C$-chain sum to zero, so a straight side at one end implies a straight side at the other. The Jacobian is then singular at the solution, and in our runs of a rectangle-to-circle design the solver failed from $N=12$ on. For a compact rectangle we therefore impose the four sides of every tile and no diagonal; congruence is then no longer imposed, and we check it after the solve: in all our rectangle designs the diagonals match to the accuracy of the sides, or to $10^{-7}$.

\subsection{Boundary-only design of compact reconfigurable patterns}\label{appendix:boundary-design}
In a compact reconfigurable pattern with straight slits, every interior vertex is the midpoint of its slit (main text Section~\ref{sec:harmonic}); combining this relation at a vertex and at the two ends of its slit makes the vertex the mean of its four diagonal neighbors, so both compact states are discretely harmonic, and so are their block centers, which are averages of vertices (Fig.~\ref{fig:SI_harmonic}). The interior is therefore fixed by the outline, and the reconfigured state follows tile by tile, each tile turned by a quarter turn. For a compact rectangle the outline has $2M+3$ free parameters: its aspect, its scale, the slides of its free vertices along the sides, and a rigid motion. We fit them by least squares to the target of the reconfigured outline, with a penalty that keeps every tile angle above $\pi/9$. On rectangle-to-target designs (circle, four ellipses, egg; $M=8$ to $24$) this takes $0.02$ to $7$\,s, against $5$, $14$ and $71$\,s for the full optimization of the rectangle-to-circle design at $M=8$, $12$ and $16$. Up to $M=16$ the fit is at most $1.2\%$ worse than that of the full optimization; at $M=20$ and $24$ the two methods end in different local minima, from $12\%$ better to $21\%$ worse.

\begin{figure}[t]
    \centering
    \includegraphics[width=\linewidth]{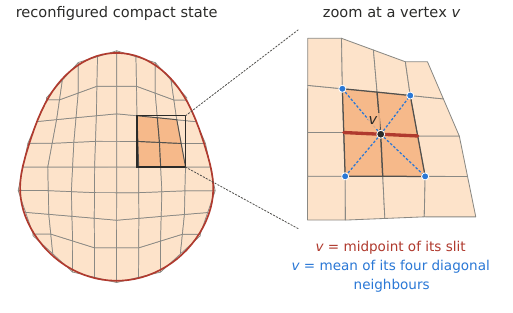}
    \caption{\textbf{Discrete harmonicity of a compact state.} The reconfigured compact state of the rectangle-to-egg design of main text Fig.~\ref{fig:demo_analysis}\textbf{a} ($M=8$; target in red), with one interior vertex $v$ enlarged: $v$ is the midpoint of its slit (red) and the mean of its four diagonal neighbors (blue), both to $10^{-14}$ of an edge.}
    \label{fig:SI_harmonic}
\end{figure}
\subsection{Initial guesses, resolution and the target curve} \label{appendix:init}
The optimization starts from a guess for both states. Following~\cite{choi2019programming}, each guess is a uniform pattern of rectangular tiles (Appendix~\ref{appendix:tess}), at any cutting angle, rescaled to the area of the target or mapped onto it by a conformal or quasi-conformal map; the two states are initialized independently.

\emph{The initial tessellation.} Which uniform pattern should we start from? Opened to the fold $\phi=\pi/4$ ($\theta=\pi/2$), a uniform $M\times M$ pattern of rectangles is a square regardless of the tiles (Appendix~\ref{appendix:uniform}); below the fold it keeps the aspect of its compact state, and above the fold the aspect is turned. A tall target such as an egg can therefore be approached either from a tall compact rectangle opened to $\theta<\pi/2$ or from a wide one opened to $\theta>\pi/2$ (Fig.~\ref{fig:demo_init}\textbf{a},\textbf{b}). The second start leads to a regular design, while from the first the optimizer distorts the tiles near the top of the egg and converges slowly: a larger cutting angle lets the tiles deform more uniformly between the states.

\emph{The resolution.} The numbers $M$ and $N$ of tiles in the two directions fix, together with the aspect of the compact rectangle, the shape of the tiles. For a compact reconfigurable rectangle-to-target design, the exact row and column spacing of Appendix~\ref{appendix:transpose} ties the compact aspect to the target: the $N$ rows of the reconfigured state are spaced by $W_C/M$ and its $M$ columns by $H_C/N$, so it is about $NW_C/M$ high and $MH_C/N$ wide, and $W_C/H_C\approx(M/N)^2\,r_{\mathrm{in}}(\Omega)$. For the wavy target of Fig.~\ref{fig:demo_init}\textbf{c},\textbf{d} ($r_{\mathrm{in}}=0.675$), a $20\times20$ pattern needs a compact rectangle taller than wide, made of short and wide tiles, and near the valleys of the target, where the boundary is steep, these tiles degenerate into triangles. A $20\times6$ pattern needs $W_C/H_C\approx7$, its tiles are tall, and they follow the steep parts of the target.

\begin{figure*}[t]
    \centering
    \includegraphics[width=\linewidth]{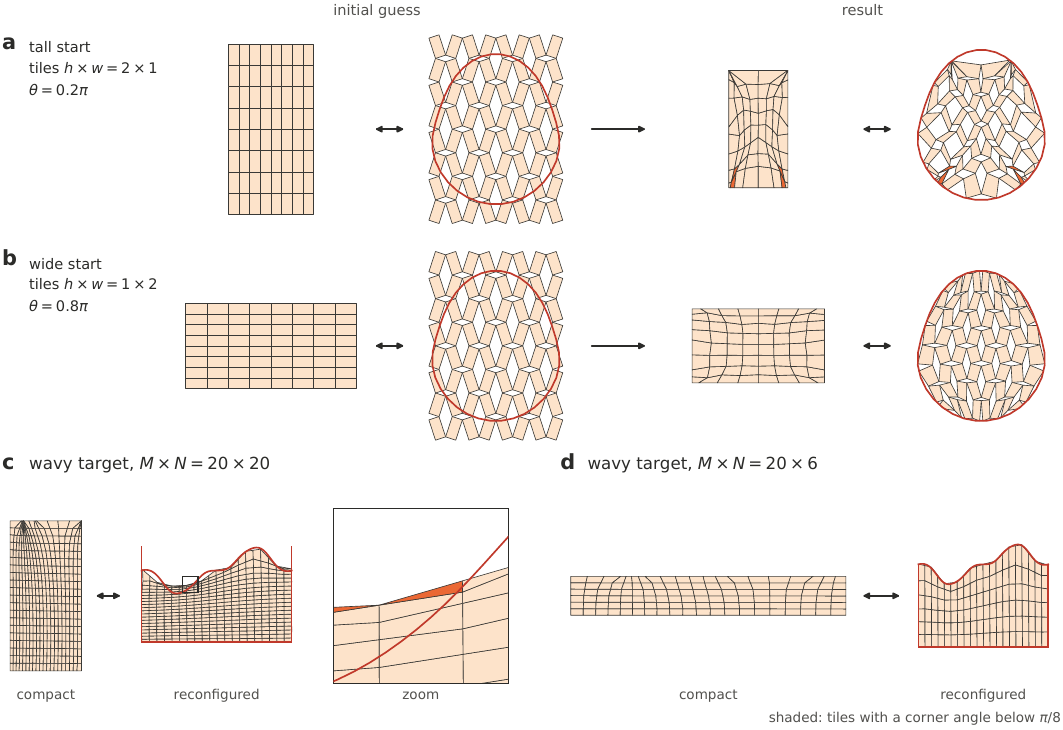}
    \caption{\textbf{The initial tessellation and the resolution.} \textbf{a},~Rectangle to egg ($M=N=8$) from a uniform pattern of $2\times1$ tiles (height $\times$ width) deployed at $\theta=0.2\pi$: initial guess (left) and result (right), each as compact $\leftrightarrow$ deployed state; red: target. The run stopped at the iteration limit with a largest edge mismatch of $3\times10^{-4}$. \textbf{b},~The same target from $1\times2$ tiles deployed at $\theta=0.8\pi$. Panels \textbf{a} and \textbf{b} are drawn at one scale. \textbf{c},~Compact reconfigurable rectangle-to-wavy design at $M\times N=20\times20$, with a zoom on the boxed valley of the target, where the tiles degenerate. \textbf{d},~The same target at $20\times6$; \textbf{c} and \textbf{d} at one scale, each design scaled to unit area. Shaded: tiles with a corner angle below $\pi/8$.}
    \label{fig:demo_init}
\end{figure*}

\emph{The target curve.} By default the boundary vertices of a pattern are matched to the whole target curve: each vertex is drawn to the nearest point of the curve, wherever it is. Alternatively, the curve can be cut into four arcs, each matched by one side of the pattern, which suits targets with corners or a clear symmetry and narrows the search, at the price of excluding designs that a different correspondence would allow. The cut changes the design (Fig.~\ref{fig:target-shape-splitting}): four equal arcs of a circle give a pattern with the symmetry of the square, longer top and bottom arcs a different pattern, and the default matching yet another.

\begin{figure}[t]
    \centering
    \includegraphics[width=\linewidth]{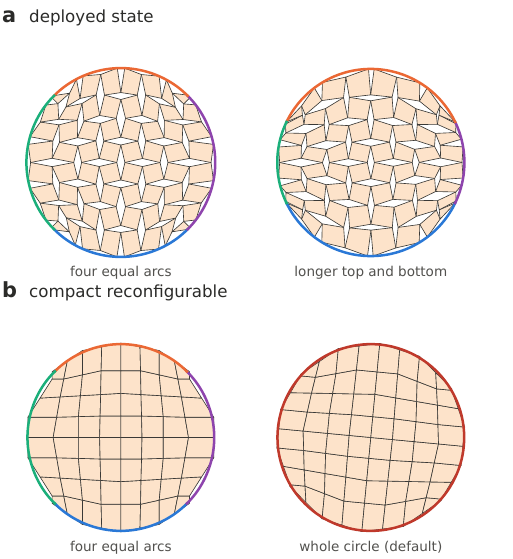}
    \caption{\textbf{Matching the target curve piecewise or as a whole} (rectangle to circle, $M=N=8$). \textbf{a},~Deployed states ($\theta=0.8\pi$) when each side of the pattern is matched to one arc of the circle (color: side of the pattern): four equal arcs (left) and longer top and bottom arcs (right). \textbf{b},~Reconfigured compact states ($\theta=\pi$) of compact reconfigurable designs with four equal arcs (left) and with the default matching of the whole boundary to the whole circle (right, red).}
    \label{fig:target-shape-splitting}
\end{figure}

\section{The lattice of blocks}\label{appendix:tess}
Section~\ref{sec:analysis} of the main text describes a pattern by the lattice of its blocks of $2\times2$ tiles. Here we derive this description: for a uniform pattern, where it is exact (Appendix~\ref{appendix:uniform}); for an optimized pattern, where the opening of each block is measured (Appendix~\ref{appendix:opening}) and the lattice law holds block by block (Appendix~\ref{appendix:lattice}); and for the curvature that follows from it (Appendix~\ref{appendix:gauss}).

\subsection{Uniform patterns}\label{appendix:uniform}
All patterns in this work have the rotating-squares topology: an $M\times N$ array of quadrilateral tiles $(i,j)$, from the bottom-left tile $(1,1)$ to the top-right tile $(M,N)$, each connected at one corner to each of its four neighbors; the vertices of tile $(i,j)$ have the global labels $m(i,j,k)$, $k=1,\dots,4$, of~\cite{dang2025shape}. The initial guesses are uniform patterns of congruent tiles, and we construct them from the lattice of tile centers, following Zheng et al.~\cite{zheng2022continuum} (Fig.~\ref{fig:make_tess}\textbf{a}).

Let the tiles be congruent parallelograms with half-edge vectors $\mathbf h_1,\mathbf h_2$, $H=[\mathbf h_1\ \mathbf h_2]$, $\det H>0$. In the compact state, their centers form the lattice $\mathbf c_{ij}-\mathbf c_{kl}=(i-k)\,\mathbf e_1+(j-l)\,\mathbf e_2$, with $\mathbf e_1=2\mathbf h_1$ and $\mathbf e_2=2\mathbf h_2$, and a \emph{block}, the cell spanned by $\mathbf e_1$ and $\mathbf e_2$ between four neighboring centers, is covered by tiles without gaps. In a deployed state, each tile moves rigidly: tile $(i,j)$ turns by $\eta_{ij}\phi$, with the checkerboard sign $\eta_{ij}=(-1)^{i+j+1}$, and each pair of neighbors stays connected at one of its two shared corners. Equating the two images of this corner for a horizontal and a vertical pair of neighbors gives, with $R(\pm\phi)=\cos\phi\,I\pm\sin\phi\,J$ and $J$ the rotation by $\pi/2$,
\begin{equation}
\label{eq:AH}
\begin{aligned}
&AH=\cos\phi\,H+\sin\phi\operatorname{cof}H,\quad\text{i.e.}\\
&A(\phi)=\cos\phi\,I+\sin\phi\,G,\qquad G=|\det H|\,(HH^{\mathsf T})^{-1},
\end{aligned}
\end{equation}
with $\operatorname{cof}H=\det(H)\,H^{-\mathsf T}$~\cite{zheng2022continuum}. Every pair of neighbors gives the same step, whatever the signs of its tiles, so every vector between tile centers is mapped by one matrix,
\begin{equation}
\label{eq:FRA}
\mathbf c'_{ij}-\mathbf c'_{kl}=F\,(\mathbf c_{ij}-\mathbf c_{kl}),\qquad F=R(\gamma)A(\phi),
\end{equation}
with a free rotation $R(\gamma)$: the deployed centers again form a parallel lattice (Fig.~\ref{fig:make_tess}\textbf{a}, middle). $F$ acts on the blocks, not on points; a point $\mathbf p$ of tile $(i,j)$ moves to $\mathbf c'_{ij}+R(\gamma+\eta_{ij}\phi)(\mathbf p-\mathbf c_{ij})$. The extra area of a block, $\det A-1$ times its compact area, is negative space, and for congruent tiles every negative space is a rhombus.

We call $\phi=\theta/2$ the \emph{opening}, with $\theta$ the cutting angle; neighboring tiles turn by $2\phi$ relative to each other. At $\phi=0$ the pattern is compact; at $\phi=\pi/2$ every tile has turned by $\pm\pi/2$, $A=G$, and the pattern closes again into the second compact state of a compact reconfigurable pattern (Fig.~\ref{fig:make_tess}\textbf{a}, right). $G$ is symmetric with $\det G=1$; with $\rho\ge1$ its larger eigenvalue the principal stretches are $\cos\phi+\rho\sin\phi$ and $\cos\phi+\rho^{-1}\sin\phi$, and the area grows by $\det A=1+\tfrac12\sin2\phi\,(\rho+\rho^{-1})$, which for every tile shape is largest at $\phi=\pi/4$, the \emph{fold}. For $h\times w$ rectangles (height $h$, width $w$), $G=\operatorname{diag}(h/w,w/h)$, a block of width $w$ and height $h$ deploys to width $w\cos\phi+h\sin\phi$ and height $h\cos\phi+w\sin\phi$, and at $\phi=\pi/2$ width and height are exchanged; this is the rotating-rectangles tessellation of our initial guesses, which our implementation reproduces to machine precision.

\begin{figure*}[t]
    \centering
    \includegraphics[width=\linewidth]{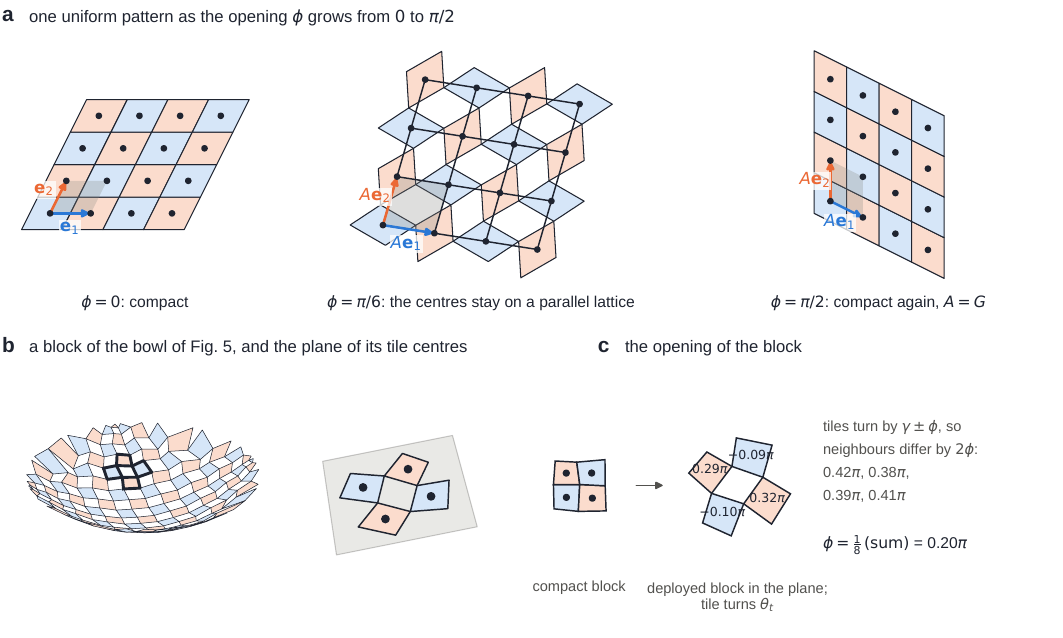}
    \caption{\textbf{The tessellation as a lattice of tile centers}, following Zheng et al.~\cite{zheng2022continuum}, \textbf{and the opening of a block in a non-uniform design.} Tiles turning by $+\phi$ are light orange and tiles turning by $-\phi$ light blue; dots: tile centers. \textbf{a},~One uniform pattern of parallelogram tiles with half-edge vectors $\mathbf h_1=(1,0)$ and $\mathbf h_2=(0.4,0.8)$ as its opening grows, drawn on one length scale ($\gamma=0$). Left: the compact state, with the lattice vectors $\mathbf e_1=2\mathbf h_1$ and $\mathbf e_2=2\mathbf h_2$ and the block they span (grey). Middle: at $\phi=\pi/6$ the centers still form a parallel lattice, spanned by $A\mathbf e_1$ and $A\mathbf e_2$ (Eq.~\eqref{eq:AH}). Right: at $\phi=\pi/2$ the pattern is compact again and $A=G$. \textbf{b},~One block of the bowl of main text Fig.~\ref{fig:multistate}, and its four tile centers with their best-fit plane (grey). \textbf{c},~The compact block and the deployed block projected into that plane, each tile labelled with its turn $\theta_t$; each of the four relative turns of neighboring tiles is $2\phi$, and $\phi=0.20\pi$ is one eighth of their sum.}
    \label{fig:make_tess}
\end{figure*}

\subsection{The opening of a block}\label{appendix:opening}
In an optimized pattern, the tiles differ and the opening varies across the sheet, so Eq.~\eqref{eq:FRA} holds only block by block, and the opening of each block has to be measured (Fig.~\ref{fig:make_tess}\textbf{b},\textbf{c}). In a 3D design, the four deployed tile centers of a block are not coplanar, and we work in their best-fit plane. For each tile we compute the rotation $\theta_t$ that takes its compact copy to its deployed copy projected into that plane (orthogonal Procrustes). The tiles turn by $\gamma+\phi$ and $\gamma-\phi$, so each of the relative rotations $|\theta_k-\theta_l|$ of the four pairs of neighboring tiles is $2\phi$, and the opening is half their mean; the common rotation $\gamma$ cancels, and for a uniform pattern the procedure gives $\phi$ exactly.

\subsection{Block centers and the lattice law}\label{appendix:lattice}
We label the blocks by continuous coordinates $(s,t)$, one unit per tile step, and let $X(s,t)$ be the compact position of the block center, the mean of its four tile centers, and $Y(s,t)$ its deployed position; $\nabla X$, $\Delta X$ and $\nabla Y$ are central differences of neighboring block centers (Fig.~\ref{fig:block_centers}). Why should Eq.~\eqref{eq:AH} survive in an optimized pattern? Its tiles are general quadrilaterals. With the edges $\mathbf e_1,\dots,\mathbf e_4$ of a tile in counter-clockwise order, $\mathbf h_1=(\mathbf e_1-\mathbf e_3)/4$ and $\mathbf h_2=(\mathbf e_2-\mathbf e_4)/4$, half the means of its pairs of opposite sides are the half edges of a parallelogram. In a compact state, neighboring tiles share an edge, and their vertex-mean centers then differ exactly by $\mathbf h_1+\mathbf h_1'$ along one lattice direction and by $\mathbf h_2+\mathbf h_2'$ along the other, so $\nabla X=2H$ up to the second difference of $H$ across the lattice. Replacing $2H$ by $\nabla X$ in Eq.~\eqref{eq:AH} block by block, with an opening $\phi(s,t)$ and a rotation $\gamma(s,t)$ that vary across the sheet, gives the lattice law, main text Eq.~\eqref{eq:lattice-law}, for $\nabla Y$ written in the tangent plane of the deployed sheet, and at $\phi=\pi/2$ the reclosure law, main text Eq.~\eqref{eq:reclosure-law}.

\begin{figure}[t!]
    \centering
    \includegraphics[width=\linewidth]{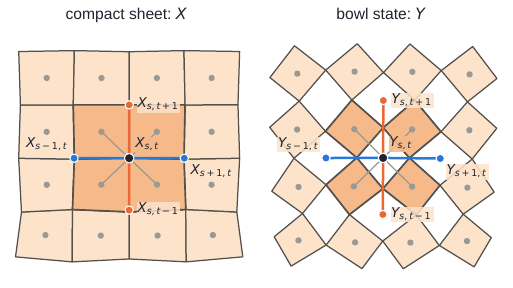}
    \caption{\textbf{Block centers and central differences}, on the sheet of main text Fig.~\ref{fig:multistate} ($M=10$): a patch of the compact sheet (left) and the same tiles in the bowl state (right, viewed along the normal of the patch). A block is a group of $2\times2$ tiles (shaded), neighboring blocks share two tiles, and the block center, $X_{s,t}$ in the compact and $Y_{s,t}$ in the deployed state, is the mean of its four tile centers (grey). The derivatives are central differences over the five-point stencil, $\partial_sX\approx\frac12(X_{s+1,t}-X_{s-1,t})$ (blue), $\partial_tX\approx\frac12(X_{s,t+1}-X_{s,t-1})$ (orange), $\nabla X=[\partial_sX\ \partial_tX]$ and $\Delta X\approx X_{s+1,t}+X_{s-1,t}+X_{s,t+1}+X_{s,t-1}-4X_{s,t}$, and likewise for $Y$, whose $3\times2$ Jacobian $\nabla Y$ spans the tangent plane of the deployed sheet.}
    \label{fig:block_centers}
\end{figure}

How well does the reclosure law hold on the designs? On the rectangle-to-target designs of Fig.~\ref{fig:targets}\textbf{a}, with one rotation $R$ per design fitted over all interior blocks, its median relative error falls from $1$--$2\%$ at $M=8$ to $0.2$--$0.3\%$ at $M=24$, as $M^{-2}$, the order of the central differences;

\subsection{The Gauss equation}\label{appendix:gauss}
Let $\theta^1,\theta^2$ be the rows of $\Phi$, read as 1-forms in the coordinates of the blocks, $\theta^a=\Phi_{a1}\,ds+\Phi_{a2}\,dt$. They are an orthonormal coframe of the deployed metric $\Phi^{\mathsf T}\Phi$, and the rotation $R(\gamma)$ changes the connection by $d\gamma$ but not the curvature. Their exterior derivatives are $d\theta^a=q_a\,ds\wedge dt$, with $q_a=\partial_s\Phi_{a2}-\partial_t\Phi_{a1}$. The curls of the rows of $\nabla X$ vanish and those of $\operatorname{cof}\nabla X$ are $-\Delta X_2$ and $\Delta X_1$, so the torsion is
\begin{equation}
\label{eq:torsion}
\begin{aligned}
\mathbf q&=\sin\phi\;J\,\Delta X\;+\;\partial_\phi\Phi\;J\,\nabla\phi,\\
\partial_\phi\Phi&=-\sin\phi\,\nabla X+\cos\phi\operatorname{cof}\nabla X .
\end{aligned}
\end{equation}
Cartan's first structure equations, $d\theta^1=-\omega\wedge\theta^2$ and $d\theta^2=\omega\wedge\theta^1$, are a linear system for the connection form $\omega=\omega_s\,ds+\omega_t\,dt$, solved by $(\omega_s,\omega_t)^{\mathsf T}=-\Phi^{\mathsf T}\mathbf q/\det\Phi$; the second, $d\omega=K\,\theta^1\wedge\theta^2=K\det\Phi\,ds\wedge dt$, is main text Eq.~\eqref{eq:gauss}.

\section{Proofs and checks for Section~\ref{sec:analysis}}\label{appendix:analysis}

\subsection{The inertia transposition}\label{appendix:inertia}
This section collects the moment identities behind the inertia transposition of main text Section~\ref{sec:harmonic}, states it precisely as Theorem~\ref{thm:inertia-free}, and proves it. For a planar region $\Omega$ with centroid $(\bar x,\bar y)$,
\begin{equation}
\bar{x} = \frac{1}{|\Omega|}\int_{\Omega} x\,dA,\qquad
\bar{y} = \frac{1}{|\Omega|}\int_{\Omega} y\,dA,
\end{equation}
the central second moments of area are
\begin{equation}
\label{eq:central-moments}
I_{xx}(\Omega)=\iint_{\Omega}(y-\bar y)^2\,dA,\qquad
I_{yy}(\Omega)=\iint_{\Omega}(x-\bar x)^2\,dA,
\end{equation}
and the inertia ratio is $r_{\mathrm{in}}(\Omega)=\sqrt{I_{xx}(\Omega)/I_{yy}(\Omega)}$, main text Eq.~\eqref{eq:rin-def}.

Compact reconfiguration turns every tile by $\pm\pi/2$ in a checkerboard. A quarter turn exchanges the coordinate axes, and therefore the two central second moments of each tile:
\begin{equation}
\label{eq:tile-transpose}
I_{xx}(\tau^2_{ij})=I_{yy}(\tau^1_{ij}),\qquad
I_{yy}(\tau^2_{ij})=I_{xx}(\tau^1_{ij}),
\end{equation}
where $\tau^k_{ij}$ is tile $(i,j)$ in state $k$, of area $a_{ij}$, and $\mathcal A=\sum_{ij}a_{ij}$. By the parallel axis theorem, the moments of a whole state are the sum of the moments of its tiles about their own centroids and of the moments of the tile centroids:
\begin{equation}
\label{eq:global-transpose}
\begin{aligned}
I_{xx}(\Omega_k)&=\textstyle\sum_{ij}I_{xx}(\tau^k_{ij})+\mathcal A\,v^2_{k,y},\\
I_{yy}(\Omega_k)&=\textstyle\sum_{ij}I_{yy}(\tau^k_{ij})+\mathcal A\,v^2_{k,x},\qquad k=1,2,
\end{aligned}
\end{equation}
where $v^2_{k,y}$ and $v^2_{k,x}$ are the area-weighted variances of the tile centroids of state $k$,
\begin{equation}
\label{eq:centroid-var}
\begin{aligned}
v^2_{k,y}&=\frac{1}{\mathcal A}\sum_{ij}a_{ij}\bigl(g^{y}_{k,ij}-\bar y_k\bigr)^2,\\
v^2_{k,x}&=\frac{1}{\mathcal A}\sum_{ij}a_{ij}\bigl(g^{x}_{k,ij}-\bar x_k\bigr)^2,
\end{aligned}
\end{equation}
with $(g^x_{k,ij},g^y_{k,ij})$ the area centroid of tile $(i,j)$ in state $k$ and $(\bar x_k,\bar y_k)$ the centroid of $\Omega_k$. 
\begin{theorem}[Inertia transposition of compact reconfigurable patterns]
\label{thm:inertia-free}
Let $(\mathcal{P}_M)_{M\ge2}$ be a sequence of compact reconfigurable patterns with planar compact states $\Omega^{(M)}_1$ and $\Omega^{(M)}_2$ of common area $\mathcal A$. Assume:
\begin{enumerate}
  \item[(H1)] \textbf{Vanishing tiles:} $\max_{ij}\mathrm{diam}(\tau_{ij})\to0$ as $M\to\infty$.
  \item[(H2)] \textbf{Convergent states:} $\chi_{\Omega^{(M)}_k}\to\chi_{\Omega_k}$ in $L^1(\mathbb{R}^2)$ for bounded Lipschitz domains $\Omega_k$, $k=1,2$.
  \item[(H3)] \textbf{Balanced centroids:} $v^2_{1,y}\,v^2_{2,y}-v^2_{1,x}\,v^2_{2,x}\to0$ as $M\to\infty$.
\end{enumerate}
Then $r_{\mathrm{in}}(\Omega_1)\,r_{\mathrm{in}}(\Omega_2)=1$, main text Eq.~\eqref{eq:rin-reciprocal}.
\end{theorem}
This is the precise form of main text Theorem~\ref{thm:inertia}. Under (H1) and (H2), the centroid variances become second moments per unit area, so (H3) is equivalent to the conclusion: the theorem places the reciprocity in the centroids of the tiles, while the tile moments transpose by themselves.

\begin{proof}
\emph{Step 1 (parallel axes).} The parallel axis theorem in each state gives Eq.~\eqref{eq:global-transpose}.

\emph{Step 2 (the tile terms vanish).} A tile of diameter $d_{ij}$ and area $a_{ij}$ has $I_{xx}(\tau^k_{ij}),I_{yy}(\tau^k_{ij})\le a_{ij}d_{ij}^2$, so the tile sums in Eq.~\eqref{eq:global-transpose} are $O\bigl(\mathcal A\,(\max_{ij}d_{ij})^2\bigr)=o(\mathcal A)$ by (H1), and
\begin{equation}
\frac{I_{xx}(\Omega_k)}{\mathcal A}=v^2_{k,y}+o(1),\qquad \frac{I_{yy}(\Omega_k)}{\mathcal A}=v^2_{k,x}+o(1).
\end{equation}

\emph{Step 3 (balance gives the reciprocal).} By Step~2,
\begin{equation}
r_{\mathrm{in}}(\Omega_1)^2\,r_{\mathrm{in}}(\Omega_2)^2=\frac{I_{xx}(\Omega_1)\,I_{xx}(\Omega_2)}{I_{yy}(\Omega_1)\,I_{yy}(\Omega_2)}=\frac{v^2_{1,y}\,v^2_{2,y}}{v^2_{1,x}\,v^2_{2,x}}+o(1),
\end{equation}
the variances being bounded below by (H2), and the right-hand side tends to $1$ by (H3). By (H2) the moments of $\Omega^{(M)}_k$ converge to those of $\Omega_k$, which gives Eq.~\eqref{eq:rin-reciprocal}.
\end{proof}

\emph{The rectangle.} A compact rectangle $R$ of width $W_C$ and height $H_C$, centered at the origin, has
\begin{equation}
\label{eq:rect-moments}
\begin{aligned}
I_{xx}(R)&=\int_{-H_C/2}^{H_C/2}\!\int_{-W_C/2}^{W_C/2} y^2\,dx\,dy=\frac{W_CH_C^3}{12},\\
I_{yy}(R)&=\int_{-H_C/2}^{H_C/2}\!\int_{-W_C/2}^{W_C/2} x^2\,dx\,dy=\frac{H_CW_C^3}{12},
\end{aligned}
\end{equation}
so $r_{\mathrm{in}}(R)=H_C/W_C$, and Eq.~\eqref{eq:rin-reciprocal} becomes $W_C/H_C=r_{\mathrm{in}}(\Omega_2)$, the aspect-ratio law of main text Eq.~\eqref{eq:inertia-law}.

\emph{Remark on (H3).} We cannot guarantee (H3) for every target: the mechanism alone does not imply it. Two examples show how it can fail, and that the failure is small (Fig.~\ref{fig:SI_hourglass}). The first is an \emph{hourglass}. Let $X$ map the index square $[0,1]^2$ onto a $W\times H$ rectangle, with its bottom and top sides parametrized by $s+\varepsilon\sin2\pi s$ and its other sides uniformly. The reclosure law (main text Eq.~\eqref{eq:reclosure-law} with $R=I$) gives the partner
\begin{equation}
\begin{aligned}
X&=Ws+\varepsilon W\,\frac{\sin2\pi s\,\cosh2\pi(t-\tfrac12)}{\cosh\pi}+\mathrm{i}Ht,\\
Z&=Hs+\mathrm{i}\Bigl(Wt+\varepsilon W\,\frac{\cos2\pi s\,\sinh2\pi(t-\tfrac12)}{\cosh\pi}\Bigr),
\end{aligned}
\end{equation}
which fills an hourglass of width $H$ and height $W-2\delta\cos(2\pi x/H)$, $x\in[-H/2,H/2]$, $\delta=\varepsilon W\tanh\pi$. It has the area $WH$ of the rectangle but the moments
\begin{equation}
\label{eq:hourglass-moments}
I_{xx}=\frac{H\,(W^3+6W\delta^2)}{12},\qquad I_{yy}=\frac{WH^3}{12}+\frac{\delta H^3}{\pi^2},
\end{equation}
so $r_{\mathrm{in}}<W/H$ for $\varepsilon>0$, while the compact rectangle keeps the aspect $W/H$. The designs follow this pair: for $W/H=2$ their compact aspect stays within $0.3\%$ of $2$ from $M=16$ on, and the aspect-ratio law is off by $2$--$3\%$. The second is a rounded triangle $r=1+\varepsilon\cos3\theta$. Its three-fold symmetry makes its moments isotropic, $r_{\mathrm{in}}=1$ for every $\varepsilon$, yet its compact rectangle is not a square: its aspect lies below the circle's by $1.0$, $3.7$ and $7.3\%$ at $\varepsilon=0.05$, $0.1$ and $0.15$, at every resolution. The shortfall grows only as $\varepsilon^2$, because the mirror $x\mapsto-x$ maps the triangle of $\varepsilon$ to that of $-\varepsilon$ and changes neither the compact aspect nor $r_{\mathrm{in}}$. In both examples the law remains a good first estimate, and the egg and the rainbow stay within $1\%$ of it.

\begin{figure}[t]
    \centering
    \includegraphics[width=\linewidth]{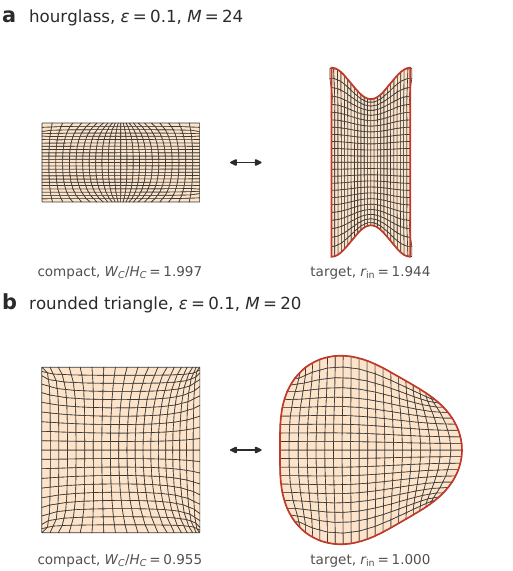}
    \caption{\textbf{Two targets on which (H3) does not hold exactly.} \textbf{a},~Rectangle-to-hourglass design ($\varepsilon=0.1$, $M=24$; target: Eq.~\eqref{eq:hourglass-moments} with $W/H=2$): compact rectangle (left) and reconfigured state on the target (right, red). Its compact aspect, $1.997$, exceeds the inertia ratio of the target, $1.944$. \textbf{b},~Rectangle-to-triangle design, $r=1+0.1\cos3\theta$, $M=20$: compact aspect $0.955$ against $r_{\mathrm{in}}=1$.}
    \label{fig:SI_hourglass}
\end{figure}

\subsection{The exact row and column spacing}\label{appendix:transpose}
Main text Section~\ref{sec:harmonic} shows that, in the continuum limit, the rows of the reconfigured compact state are spaced like the columns of the compact state and vice versa. For a compact rectangle, the mean spacing is exact at every resolution; the statement concerns the lattice of tile centers, not the shape it fills. To each tile $(i,j)$ of the reconfigured compact state we attach its \emph{vertex-mean center}
\begin{equation}
\label{eq:cD-def}
c^{D}_{i,j}=\tfrac14(p_1+p_2+p_3+p_4),
\end{equation}
the average of its four corners, which, unlike the area centroid, depends on the vertices alone. We measure the row and column links of the lattice of centers by their mean vertical and horizontal projections,
\begin{equation}
\label{eq:hD-def}
\begin{aligned}
h^{D,y}_{j}&=\frac{1}{M}\sum_{i=1}^{M}\bigl|y(c^{D}_{i,j+1})-y(c^{D}_{i,j})\bigr|,\\
h^{D,x}_{i}&=\frac{1}{N}\sum_{j=1}^{N}\bigl|x(c^{D}_{i+1,j})-x(c^{D}_{i,j})\bigr| .
\end{aligned}
\end{equation}

\begin{theorem}[Exact row and column spacing]
\label{thm:transpose}
Assume that (i) the compact state is an $M\times N$ rectangle of width $W_C$ and height $H_C$, so that every row of tiles telescopes to $W_C$ and every column to $H_C$; and (ii) the reconfigured compact state is obtained by turning every tile by exactly $\pm\pi/2$, with the sign alternating in a checkerboard and the tiles remaining edge to edge. Then
\begin{equation}
h^{D,y}_{j}=\frac{W_C}{M},\qquad h^{D,x}_{i}=\frac{H_C}{N},
\label{eq:grid-transpose}
\end{equation}
whatever the target of the reconfigured state.
\end{theorem}

\begin{proof}
We prove the row identity; the column identity follows by exchanging $i$ and $j$. The reconfigured pattern is built of \emph{layers} (Fig.~\ref{fig:SI_transpose}): layer $j$, $j=1,\dots,N+1$, is the chain of $M+1$ reconfigured vertices shared by tile rows $j-1$ and $j$, from left to right, layer $1$ being the bottom boundary and layer $N+1$ the top. Let $y^{(j)}_k$ be the $y$-coordinate of its $k$-th vertex. Tile $(i,j)$ has two corners on layer $j$, at positions $i$ and $i+1$, and two on layer $j+1$, so
\begin{equation}
y(c^{D}_{i,j}) = \tfrac{1}{4}\bigl(y^{(j)}_i + y^{(j)}_{i+1} + y^{(j+1)}_i + y^{(j+1)}_{i+1}\bigr).
\end{equation}
Summed over the row, every vertex of layer $j+1$ cancels, since it is a top corner of tile $(i,j)$ and a bottom corner of tile $(i,j+1)$, and only layers $j$ and $j+2$ remain:
\begin{equation}
\label{eq:layer-identity}
\begin{aligned}
&\sum_{i=1}^{M}\bigl[y(c^{D}_{i,j+1})-y(c^{D}_{i,j})\bigr]
=\frac14\Bigl[(y^{(j+2)}_1-y^{(j)}_1)\\
&\quad+2\sum_{k=2}^{M}(y^{(j+2)}_k-y^{(j)}_k)+(y^{(j+2)}_{M+1}-y^{(j)}_{M+1})\Bigr].
\end{aligned}
\end{equation}
Under the quarter turn, the vertical edges of the tiles of rows $j$ and $j+1$ in the reconfigured state are the horizontal edges of the same tiles in the compact state, so their $y$ projections are the $x$ projections of those compact edges. In the compact rectangle, the horizontal edges of a row of tiles telescope to $W_C$ along each of the three layers bounding rows $j$ and $j+1$, so the bracket is $4W_C$: $2W_C$ from the middle layer, counted twice, and $W_C$ from each outer layer. The signed sum is therefore $W_C$, and so is the sum of the absolute projections in Eq.~\eqref{eq:hD-def} whenever every link of the row points the same way, as in all our designs.
\end{proof}

\begin{figure}[t]
    \centering
    \includegraphics[width=\linewidth]{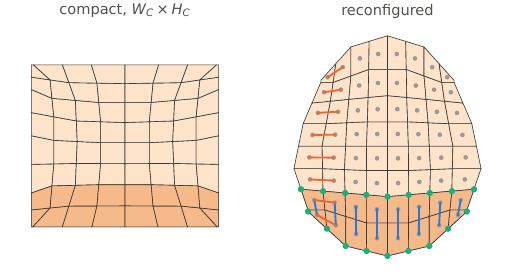}
    \caption{\textbf{The exact row and column spacing} on the compact reconfigurable rectangle-to-egg design ($M=N=8$), the reconfigured state turned so that every tile turns by exactly $\pm\pi/2$. Shaded: tile rows $1$ and $2$ in both states. In the reconfigured state, blue segments join the vertex-mean centers (grey) of rows $1$ and $2$, and their vertical projections sum to the compact width $W_C$; orange segments join those of columns $1$ and $2$, and their horizontal projections sum to the compact height $H_C$. Green: the vertices of layers $1$ and $3$, the only ones left in the row sum, Eq.~\eqref{eq:layer-identity}.}
    \label{fig:SI_transpose}
\end{figure}

\begin{cor}[Arithmetic progression of the row means]
\label{cor:AP}
Under the assumptions of Theorem~\ref{thm:transpose} with $M=N$, the row-mean heights $Y_j=\frac1M\sum_{i=1}^{M} y(c^D_{i,j})$ form an arithmetic progression of step $W_C/M$, and the column-mean abscissae $X_i=\frac1M\sum_{j=1}^{M} x(c^D_{i,j})$ one of step $H_C/M$. Their variances about their means are
\begin{equation}
\label{eq:VY}
V_Y=\frac{W_C^2}{12}\Bigl(1-\frac{1}{M^2}\Bigr),\qquad V_X=\frac{H_C^2}{12}\Bigl(1-\frac{1}{M^2}\Bigr),
\end{equation}
so that, at every resolution,
\begin{equation}
\label{eq:VY-ratio}
\sqrt{V_Y/V_X}=W_C/H_C .
\end{equation}
\end{cor}

\begin{proof}
Summing Eq.~\eqref{eq:layer-identity} with its signs, $M(Y_{j+1}-Y_j)=W_C$, and $M$ equally spaced values with step $d$ have the variance $d^2(M^2-1)/12$; likewise for $X_i$.
\end{proof}

Equation~\eqref{eq:VY-ratio} has the form of the aspect-ratio law: by Step~2 of the proof of Theorem~\ref{thm:inertia-free}, $r_{\mathrm{in}}(\Omega_2)^2$ is, for small tiles, the ratio $v^2_{2,y}/v^2_{2,x}$ of the area-weighted centroid variances~\eqref{eq:centroid-var}, while $V_Y$ and $V_X$ are the unweighted variances of the row and column means of the vertex-mean centers. For the compact rectangle $v^2_{1,y}/v^2_{1,x}\to(H_C/W_C)^2$, so (H3) says exactly that the two ratios of the reconfigured state agree in the limit, $v^2_{2,y}/v^2_{2,x}\to V_Y/V_X$: the scatter of the centers about their row and column means must not change the ratio of the variances.

\begin{figure*}[t]
    \centering
    \includegraphics[width=\linewidth]{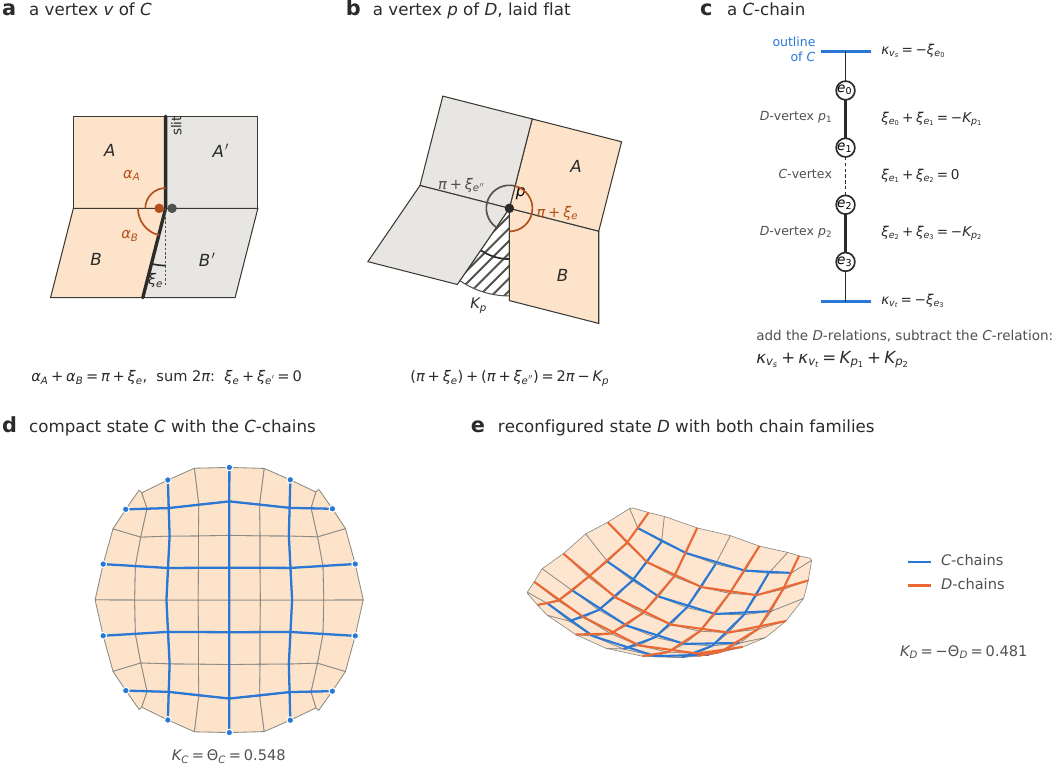}
    \caption{\textbf{Hinge chains of a compact reconfigurable 2D-to-3D design.} \textbf{a}, An interior vertex $v$ of the planar compact state $C$ (schematic, angles exaggerated): the tiles $A$, $B$ of the hinge $e$ (orange dot) and $A'$, $B'$ of the hinge $e'$ (grey dot); the slit through $v$ (bold) bends by $\xi_e$ at $e$. \textbf{b}, An interior vertex $p$ of the reconfigured state $D$, its four tile corners laid flat: the corners of $e$ (peach) and of a second hinge $e''$ (grey) leave the angle defect $K_p$ open (hatched). \textbf{c}, A $C$-chain, whose hinges are linked alternately through $D$-vertices (solid) and $C$-vertices (dashed); adding the relations of \textbf{b} and subtracting those of \textbf{a} leaves the turning of the outline of $C$ at the two ends. \textbf{d}, A compact reconfigurable bowl design ($M=8$): the compact state with its $C$-chains (dots: chain ends). \textbf{e}, Its reconfigured state with the $C$-chains (blue) and the $D$-chains, whose ends lie on the boundary of $D$ (orange). The total curvature on the $C$-chains equals the total turning of the outline of $C$ at their ends, $K_C=\Theta_C$, and $K_D=-\Theta_D$ (Theorem~\ref{thm:hinge}).}
    \label{fig:SI_chains}
\end{figure*}

\subsection{Hinge chains}\label{appendix:chains}
This section proves main text Eq.~\eqref{eq:chain}, the exact discrete form of the reclosed Gauss equation~\eqref{eq:reclosed-gauss}. Consider a compact reconfigurable pattern with a planar compact state $C$ and a reconfigured compact state $D\subset\mathbb R^3$. A \emph{hinge} $e$ is a pair of tile corners of tiles $A$ and $B$ that coincide in the deployed state, and hence in $C$ and $D$. Its \emph{slit excess} $\xi_e=\alpha_A+\alpha_B-\pi$, with $\alpha_A,\alpha_B$ the tile angles at the hinge, is the bend of the slit at $e$ (Fig.~\ref{fig:SI_chains}\textbf{a}); it depends only on the tile shapes, and $\xi_e=0$ for a straight slit. The curvature of $D$ at an interior vertex $p$ is the angle defect $K_p=2\pi-\sum(\text{tile angles at }p)$, and the turning of a boundary at a vertex $v$ is $\kappa_v=\pi-\sum(\text{tile angles at }v)$.

For even $M$ and $N$ the hinges are arranged as follows, as we verified by enumeration for $4\times4$ to $16\times16$ patterns: every interior vertex of $C$ and of $D$ hosts exactly two hinges; linking hinges alternately through shared $C$- and $D$-vertices partitions them into $(M-1)+(N-1)$ chains, one per interior grid line; both ends of a chain lie on the same boundary, which divides the chains into $C$-chains (ends on the outline of $C$) and $D$-chains (ends on the boundary of $D$); at a chain end only the two corners of the end hinge meet; and every interior vertex of $D$ lies on exactly one chain (Fig.~\ref{fig:SI_chains}\textbf{d},\textbf{e}).

\begin{theorem}[Hinge chains]
\label{thm:hinge}
Under these conditions:
\begin{enumerate}
\item[(i)] $\xi_e+\xi_{e'}=0$ for the two hinges $e,e'$ at an interior vertex of $C$;
\item[(ii)] $\xi_e+\xi_{e''}=-K_p$ for the two hinges $e,e''$ at an interior vertex $p$ of $D$;
\item[(iii)] for a $C$-chain with end vertices $v_s,v_t$ on the outline of $C$, $\sum_{p}K_p=\kappa^C_{v_s}+\kappa^C_{v_t}$; for a $D$-chain with end vertices on the boundary of $D$, $\sum_{p}K_p=-(\kappa^D_{v_s}+\kappa^D_{v_t})$.
\end{enumerate}
Consequently, with $K_C$ and $K_D$ the total curvature on the interior $D$-vertices of the $C$-chains and of the $D$-chains, and $\Theta_C$, $\Theta_D$ the total turning of the outline of $C$ at the ends of $C$-chains and of the boundary of $D$ at the ends of $D$-chains,
\begin{equation}
\label{eq:boundary-form}
K_C=\Theta_C,\quad K_D=-\Theta_D,\quad K=\Theta_C-\Theta_D .
\end{equation}
\end{theorem}

\begin{proof}
(i) Four tile corners meet at an interior vertex of the planar state $C$, grouped into the two hinges, so $2\pi=(\xi_e+\pi)+(\xi_{e'}+\pi)$ (Fig.~\ref{fig:SI_chains}\textbf{a}). (ii) At an interior vertex $p$ of $D$ the same count gives $2\pi-K_p=(\xi_e+\pi)+(\xi_{e''}+\pi)$ (Fig.~\ref{fig:SI_chains}\textbf{b}). (iii) Number the hinges of a chain $e_0=e_s,e_1,\dots,e_{2n-1}=e_t$. Along a $C$-chain consecutive hinges are linked alternately through an interior $D$-vertex and an interior $C$-vertex, beginning and ending with a $D$-link, because the $C$-vertices of the end hinges lie on the outline; call the $D$-vertices $p_1,\dots,p_n$. Adding the $n$ relations (ii), $\xi_{e_{2m}}+\xi_{e_{2m+1}}=-K_{p_{m+1}}$, and subtracting the $n-1$ relations (i), $\xi_{e_{2m-1}}+\xi_{e_{2m}}=0$, telescopes to $\xi_{e_s}+\xi_{e_t}=-\sum_mK_{p_m}$ (Fig.~\ref{fig:SI_chains}\textbf{c}). At an end vertex only the two corners of the end hinge meet, so $\kappa_v=-\xi_e$, which gives the $C$-chain statement. A $D$-chain begins and ends with a $C$-link, and the same telescoping gives $\xi_{e_s}+\xi_{e_t}=\sum_mK_{p_m}$. Summing over each family gives the first two relations of Eq.~\eqref{eq:boundary-form}, and $K=K_C+K_D$ because every interior vertex of $D$ lies on exactly one chain.
\end{proof}

On the bowl of Fig.~\ref{fig:SI_chains}\textbf{d},\textbf{e} ($M=8$), $K=1.029$, with $K_C=\Theta_C=0.548$ and $K_D=-\Theta_D=0.481$ to within $4\times10^{-6}$; the residual is the tile-angle mismatch allowed by the solver tolerance.

\begin{cor}[Rectangular compact outlines]
\label{cor:rect-outline}
If the compact state $C$ is a rectangle, then $K_C=0$. In particular, if every interior vertex of the reconfigured state $D$ has $K_p>0$, as for a pattern with all its vertices on a sphere, then $C$ cannot be a rectangle.
\end{cor}

\begin{proof}
The ends of $C$-chains lie on the sides of the outline, not at its corners, and the sides of a rectangle do not turn, so $\Theta_C=0$ and $K_C=0$; a sum of positive angle defects cannot vanish.
\end{proof}

Straight slits ($\xi_e=0$ for all $e$) give $K_p=0$ at every vertex by (ii), so a rigidly reconfigurable sheet reconfigures only into a developable state.

\section{Additional results}\label{appendix:additionalresults}

\begin{figure*}[t!]
    \centering
    \includegraphics[width=\linewidth]{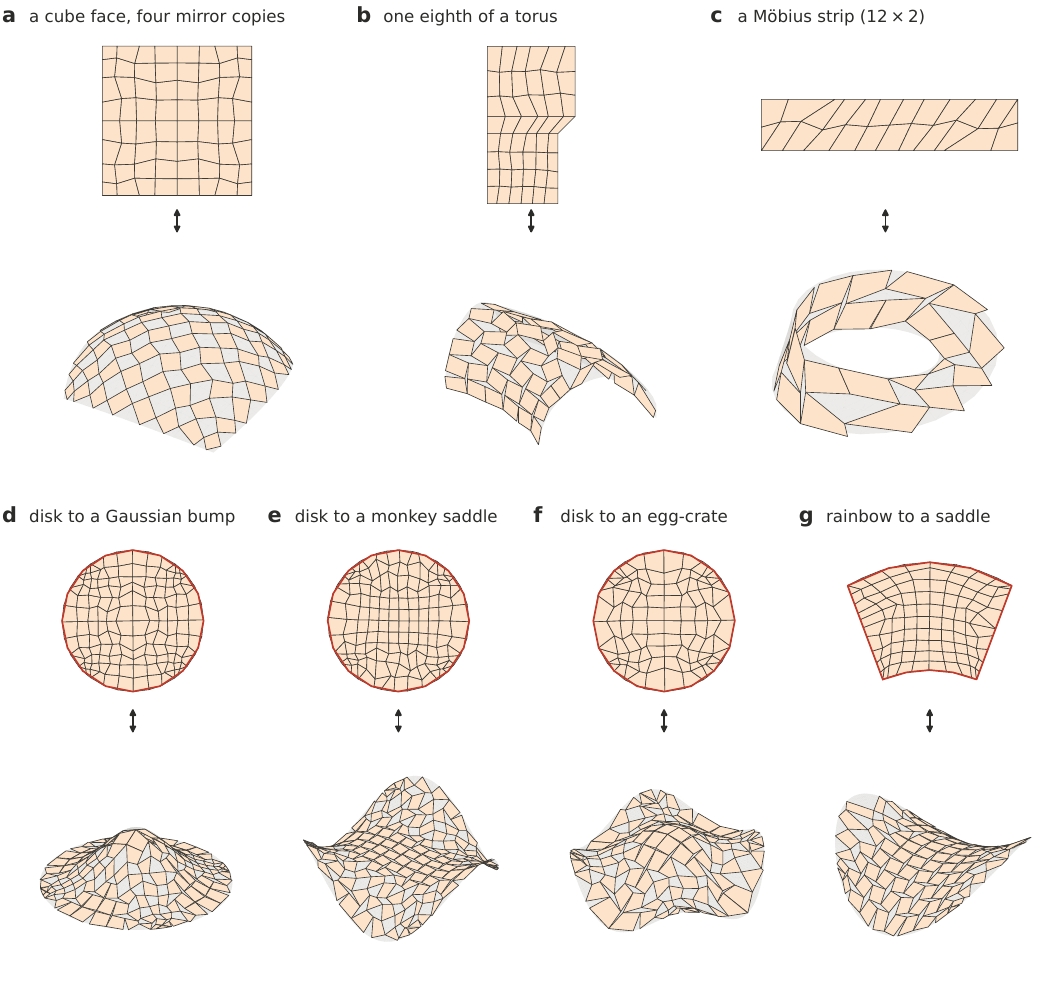}
    \caption{\textbf{2D-to-3D kirigami structures designed by our framework.} Compact state (top) and deployed state (bottom; grey: target surface). \textbf{a},~One face of the cube-to-sphere structure: four mirror copies of a square-to-spherical-patch design, deployed onto $1/6$ of the sphere. \textbf{b},~A compact state of two rectangles joined by a trapezoid, deployed onto one eighth of the upper half of a torus. \textbf{c},~A $12\times2$ strip deployed onto a M\"obius strip. \textbf{d}--\textbf{g},~Curved compact outlines (red: target curve) deployed onto graph surfaces over a disk: \textbf{d},~a disk onto the Gaussian bump $z=0.6\,e^{-4(x^2+y^2)}$; \textbf{e},~onto the monkey saddle $z=0.35\,(x^3-3xy^2)$; \textbf{f},~onto the egg-crate $z=0.28\cos\pi x\cos\pi y$; \textbf{g},~a rainbow onto the saddle $z=0.4\,(x^2-y^2)$ ($M=N=12$ in \textbf{d},\textbf{e} and $10$ in \textbf{f},\textbf{g}). In \textbf{e} and \textbf{g} every compact tile angle is kept in $[\pi/12,11\pi/12]$, except at the four corners of the pattern, which lie on the smooth outline; without this bound a few tiles nearly degenerate. Every edge length is matched to within $2\times10^{-4}$ of the mean tile edge.}
    \label{fig:SI_results_2Dto3D}
\end{figure*}

\begin{figure*}[t!]
    \centering
    \includegraphics[width=\linewidth]{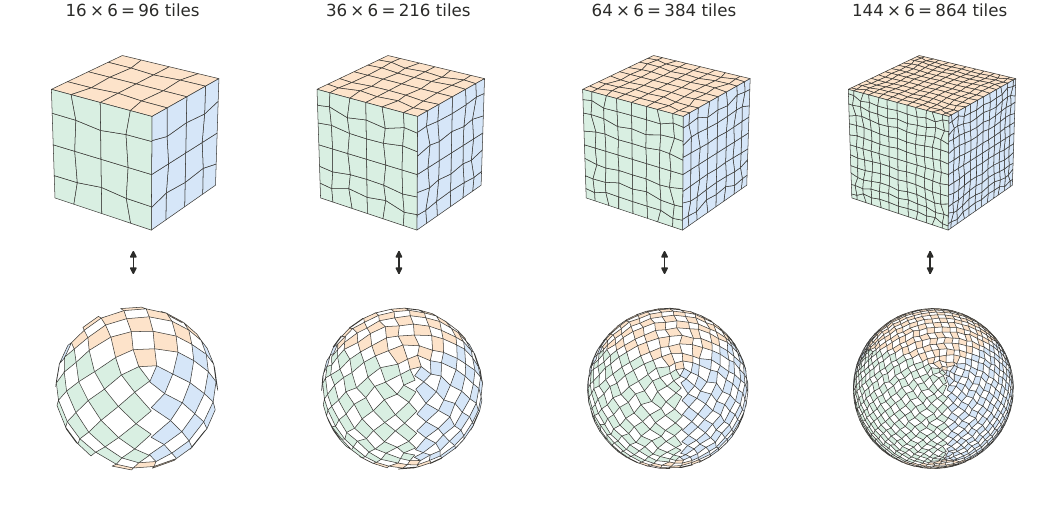}
    \caption{\textbf{Cube-to-sphere structures at four resolutions}, with $16$, $36$, $64$ and $144$ tiles per face ($96$, $216$, $384$ and $864$ tiles in all): compact cube (top) and deployed sphere (bottom). Each face consists of four mirror copies of one square-to-spherical-patch design; opposite faces share a color.}
    \label{fig:cube-to-sphere-resolution}
\end{figure*}

\subsection{3D kirigami structures}
In the main text, we have presented various 2D-to-3D kirigami structures designed by our framework. Additionally, the 2D-to-3D kirigami design can also be utilized to simplify 3D-to-3D kirigami design problems. As described in the main text, one way to design 3D-to-3D kirigami structures is by assembling multiple copies of 2D-to-3D kirigami structures.

In Fig.~\ref{fig:SI_results_2Dto3D}\textbf{a}, we show how we can create a 2D-to-3D kirigami structure that morphs from a square to a prescribed portion of a spherical cap. Here, we enforce both the surface matching and the boundary shape matching constraints in the deployed shape, thereby ensuring that the boundary points of the deployed kirigami structure lie on the target boundary. Then, by duplicating the design in both the contracted and deployed states, we can create a kirigami structure that morphs from a larger 2D square to $1/6$ of the sphere. This structure can then be further duplicated to form the cube-to-sphere example shown in the main text. Next, we consider approximating a torus shape. Since we implement optimization on both the initial pattern and the deployed pattern, we can easily control the shapes of both patterns. As shown in Fig.~\ref{fig:SI_results_2Dto3D}\textbf{b}, we can construct a 2D shape comprised of two rectangles connected by a trapezoid where the side edge has a slope of 1, and it can be deployed to be one patch of a torus, which is exactly one-eighth of the upper half of a torus. As long as the boundary condition is well satisfied, we can reflect this piece to get the whole torus as shown in the main text. Finally, a $12\times2$ strip can be deployed onto a M\"obius strip (Fig.~\ref{fig:SI_results_2Dto3D}\textbf{c}). The compact outline need not be a polygon either: as in the 2D designs, its boundary can be held on a planar target curve while the deployed state matches a surface whose curvature changes sign. A disk deploys onto a Gaussian bump, curved like a sphere at the top and like a saddle on its skirt, onto a monkey saddle and onto an egg-crate surface, and a rainbow-shaped compact state deploys onto a saddle (Fig.~\ref{fig:SI_results_2Dto3D}\textbf{d}--\textbf{g}).

Note that the above-mentioned strategies can be naturally applied for creating kirigami structures with different resolutions. In Fig.~\ref{fig:cube-to-sphere-resolution}, we show multiple cube-to-sphere deployable structures with different resolutions, from which we can see that all of them satisfy the desired shape morphing effects.

In Fig.~\ref{fig:SI_results_3Dto3D}, we further show several examples of 3D-to-3D compact reconfigurable kirigami structures designed by our framework. As we can see, even with the same target shape for the first compact state (a spherical cap), we can obtain different kirigami structures achieving highly different target shapes in the second (reconfigured) compact state, including saddle, catenoid, vase, and partial torus. This demonstrates the great flexibility of our unified design framework.

\begin{figure*}[t]
    \centering
    \includegraphics[width=\linewidth]{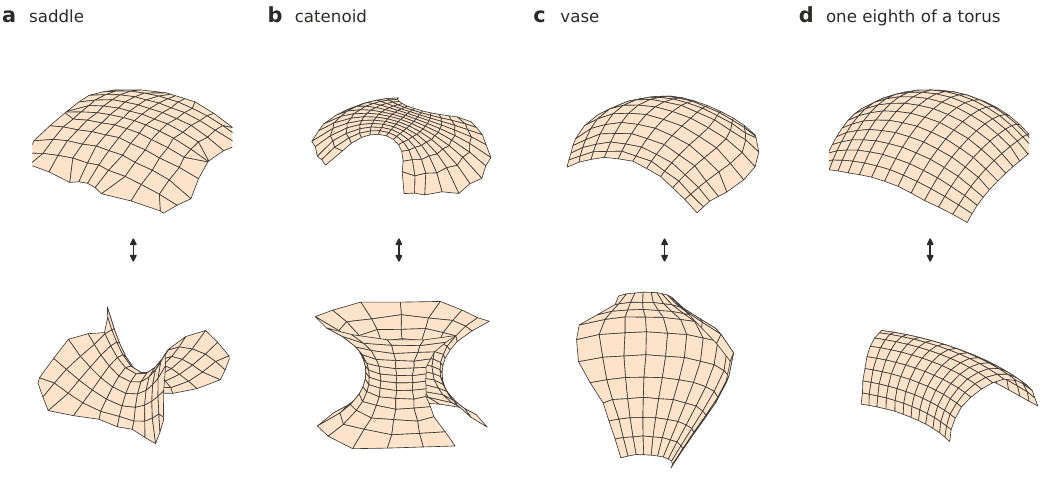}
    \caption{\textbf{3D-to-3D compact reconfigurable kirigami structures designed by our framework.} For each example, the first compact state approximates a spherical cap (top), and the reconfigured compact state (bottom) approximates another target shape: \textbf{a},~a saddle; \textbf{b},~a catenoid; \textbf{c},~a vase; \textbf{d},~one eighth of a torus.}
    \label{fig:SI_results_3Dto3D}
\end{figure*}

\subsection{Extreme shape and size changes}
Besides, note that kirigami structures are known for their ability to undergo large changes in shape and size during deployment. Using our proposed framework, we can easily design kirigami structures with controllable extreme changes in shape and size. In Fig.~\ref{fig:rec-to-square-with-ratio}, we create kirigami patterns with different square-to-circle and rectangle-to-circle effects. Specifically, besides achieving a square-to-circle shape change, one can also achieve a rectangle-to-circle shape change with different prescribed width-to-height aspect ratios (2:1, 4:1, 8:1) for the rectangular compact state.

\begin{figure*}[t]
    \centering
    \includegraphics[width=\linewidth]{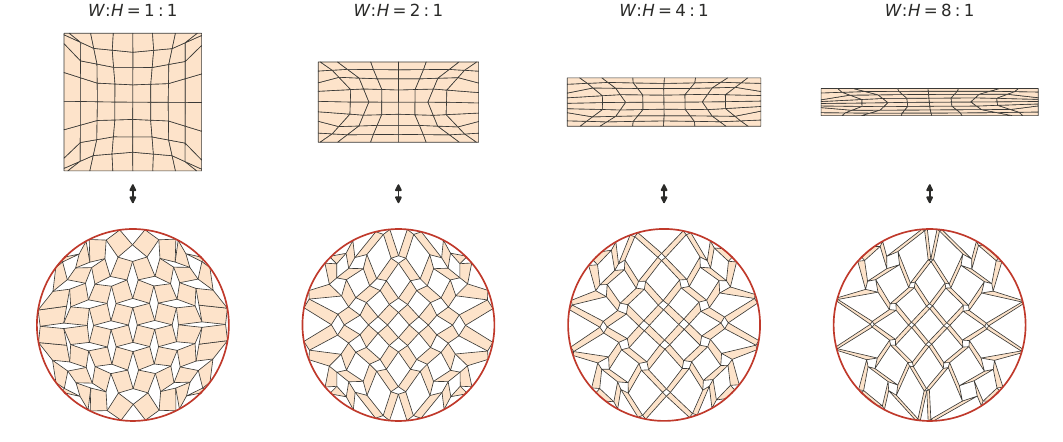}
    \caption{\textbf{Rectangle-to-circle structures with prescribed compact aspect ratios} $W\!:\!H=1{:}1$, $2{:}1$, $4{:}1$ and $8{:}1$ ($M=N=8$): compact state (top) and deployed state on the target circle (bottom, red), all at one scale.}
    \label{fig:rec-to-square-with-ratio}
\end{figure*}

Alternatively, we can also design rigidly deployable, compact reconfigurable kirigami patterns with different extreme target shapes. In Fig.~\ref{fig:targets}\textbf{a}, we design kirigami patterns that morph from a compact rectangle into target shapes with different inertia ratios at the reconfigured compact state, including a circle and ellipses with different aspect ratios (we designed seven targets; the ellipses $3{:}4$ and $1{:}2$, the transposes of $4{:}3$ and $2{:}1$, are not shown). Fig.~\ref{fig:targets}\textbf{b} shows free-to-target designs, whose first compact state is free, at three resolutions. It can be observed that our unified design framework is capable of achieving all these different target shapes.

Altogether, the above experiments demonstrate the great flexibility of our proposed framework in achieving various extreme changes in size and shape.

\begin{figure*}[t]
    \centering
    \includegraphics[width=\linewidth]{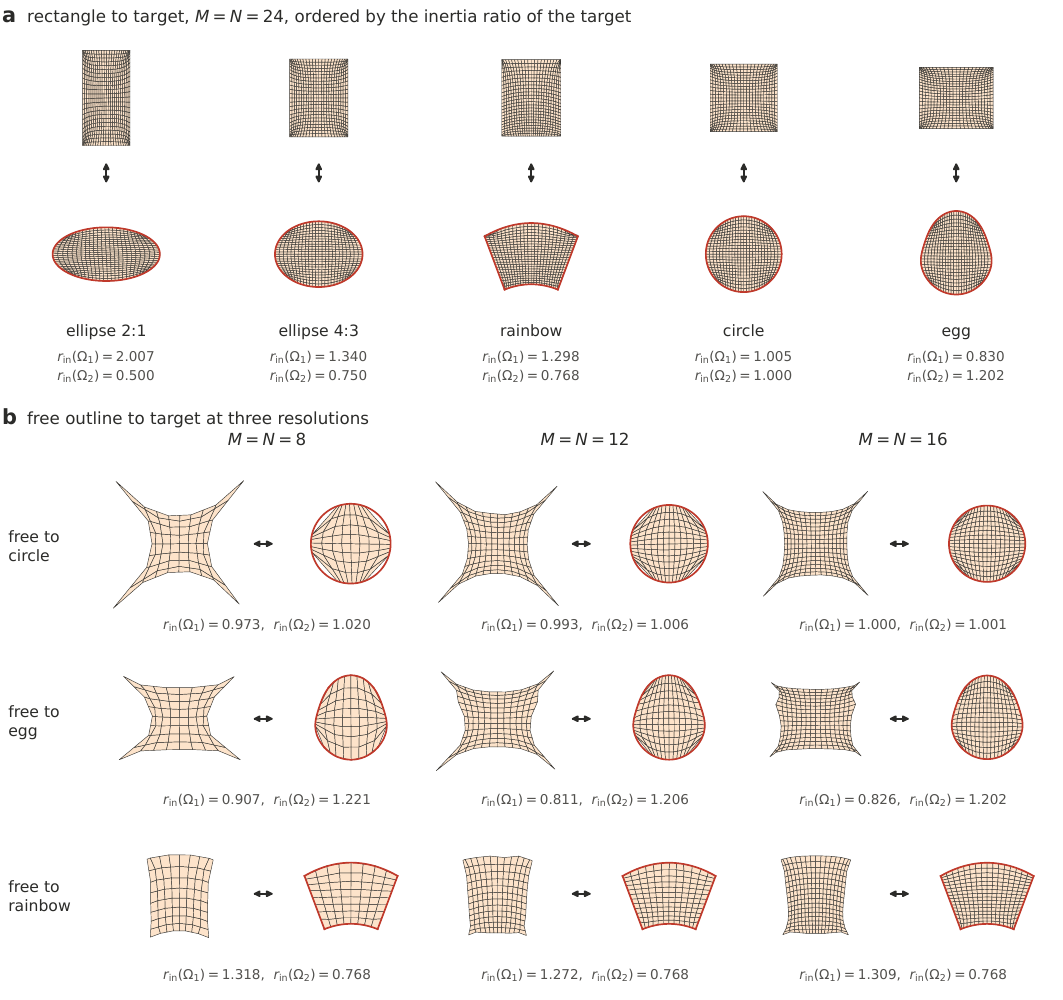}
    \caption{\textbf{2D rigidly deployable, compact reconfigurable kirigami structures obtained by our design framework.} \textbf{a}, Rectangle-to-target designs with $M=N=24$: the first compact state is a rectangle (top) and the reconfigured compact state matches the target (bottom); the targets are ordered by their inertia ratio $r_{\mathrm{in}}=\sqrt{I_{xx}/I_{yy}}$. \textbf{b}, Free-to-target designs at $M=N=8$, $12$ and $16$: the first compact state is free (left) and the reconfigured compact state matches the target (right). Every design is scaled to unit area and drawn at one scale; red: target boundaries; below each design: the inertia ratios of its two compact states $\Omega_1$ and $\Omega_2$ (main text Fig.~\ref{fig:demo_analysis}\textbf{b}).}
    \label{fig:targets}
\end{figure*}

\end{document}